\documentclass[]{interact}

\usepackage{epstopdf}
\usepackage[caption=false]{subfig}

\usepackage[numbers,sort&compress]{natbib}
\bibpunct[, ]{[}{]}{,}{n}{,}{,}

\usepackage{hyperref}

\theoremstyle{plain}
\newtheorem{theorem}{Theorem}[section]

\theoremstyle{definition}

\theoremstyle{remark}

\usepackage{booktabs}
\usepackage{siunitx}
\usepackage{multirow}
\usepackage{float}
\usepackage{graphicx}

\begin{document}

\articletype{ARTICLE TEMPLATE}

\title{Optimal Control of Reserve Asset Portfolios for Stablecoins}

\author{
\name{Alexander Hammerl\textsuperscript{a}\thanks{CONTACT Alexander Hammerl. Email: asiha@dtu.dk}}
\affil{\textsuperscript{a}Technical University of Denmark, Department of Management Engineering, Bygningstorvet 358, Kgs. Lyngby, 2800, Denmark.}
}

\maketitle

\begin{abstract}

Stablecoins promise par convertibility, yet issuers must balance immediate liquidity against yield on reserves to keep the peg credible. We study this treasury problem as a continuous-time control task with two instruments: reallocating reserves between cash and short-duration government bills, and setting a spread fee for either minting or burning the coin. Mint and redemption flows follow mutually exciting processes that reproduce clustered order flow. Peg deviations arise when immediate cash coverage is insufficient relative to outstanding supply, and the market price relaxes toward this liquidity-coverage fair value. We develop a stochastic model predictive control framework that incorporates moment closure for event intensities. Using Pontryagin's Maximum Principle, we show that the optimal reallocation control exhibits a soft-thresholding structure: no rebalancing occurs when the shadow-cost differential lies within a deadzone set by transaction costs, and reallocation scales linearly beyond that threshold up to a capacity-imposed saturation limit. Introducing settlement windows leads to a sampled-data implementation with a simple threshold (soft-thresholding) structure for rebalancing. We also establish a monotone stress-response property: as expected outflows intensify or windows lengthen, the optimal policy shifts predictably toward cash. In simulations covering various stress test scenarios, the controller preserves most bill carry in calm markets, builds cash quickly when stress emerges, and avoids unnecessary rotations under transitory signals. The proposed policy is implementation-ready and aligns naturally with operational cut-offs. Our results translate empirical flow risk into auditable treasury rules that improve peg quality without sacrificing avoidable carry.

\end{abstract}


\section{Introduction}

Stablecoins have become an indispensable component of the infrastructure of crypto-asset markets and an increasingly relevant settlement rail beyond them. Their defining promise is par convertibility to a more traditional asset, usually a fiat currency or a commodity, on demand. Fulfilling that promise requires an issuer to satisfy redemptions promptly while earning sufficient carry on reserves to operate sustainably. Regulatory agencies typically restrict issuers to a narrow range of permissible reserve assets. For example, the recently enacted U.S.\ GENIUS Act \cite{GENIUSAct2025} requires that stablecoin issuers pegged to the US dollar hold currency reserves exclusively in cash, short-term US government bonds, or overnight repurchase agreements thereof. In operational terms, this creates a treasury policy problem: how much of the reserve to hold as immediately available cash versus short-duration government bills, and how to set mint and burn spreads that help absorb order-flow imbalances without creating persistent departures of the secondary-market price from one. This paper develops a continuous-time framework that treats reserve allocation and fee setting as joint controls under stochastic mint and redemption flows. The objective is to minimize peg deviations and trading frictions while preserving yield, with policies designed to be implementable within daily treasury workflows.

Closest to the present paper is \cite{GoelLewrickAgarwal2024}, who develop a
discrete-time dynamic model of a fiat-backed stablecoin issuer that allocates
reserves between cash and bonds to maximize discounted cash flows. They show
that liquidity and capital requirements jointly reduce default probability and the
expected price impact of forced bond sales, and provide a calibrated mapping
from prudential targets to requirement levels. Our work differs in that we adopt the issuer's operational perspective rather than solving for a regulatory optimum, and in the modelling assumptions on both market dynamics and the control framework. While the framework of \cite{GoelLewrickAgarwal2024}
informs the calibration of prudential requirements, ours translates such requirements into auditable, real-time treasury rules.
While fiat-backed stablecoins operate outside the full decentralization paradigm, their treasury management problem intersects with broader dynamics in decentralized finance (DeFi). A literature review by \cite{meyer2022defi} identifies stablecoins as a central class of "financial tokens" in DeFi, emphasizing that even fully collateralized designs can deviate from their peg, exhibit state-dependent volatility, and depend critically on liquidity provision and external data feeds. This positioning situates issuer reserve policies alongside on-chain mechanisms, market microstructure effects, and inter-protocol linkages that can amplify or dampen peg pressures.

Related market-design research links validator and block-builder incentives to liquidity and price stability in on-chain asset markets. \cite{oez2023waitinggames} show that block-production timing strategies can yield measurable gains in maximal extractable value (MEV) without necessarily harming consensus stability. For pegged digital assets that rely on on-chain liquidity pools, MEV-maximizing behaviors among market participants can affect execution quality for redemptions and rebalancing, introducing another stochastic element to the issuer’s reserve-management problem.

\cite{cumming2025cryptofunds} show that well-connected crypto funds can boost token valuation and long-run performance, suggesting that similar strategic engagement with liquidity providers and market-makers could strengthen stablecoin reserve-management.

Recent empirical work clarifies why issuers must carefully manage both their mint-and-burn policies and reserve assets, even for fully collateralized digital currencies. \cite{momtaz2021pricing} documents that while the mean cryptocurrency issued in an initial coin offering (ICO) can deliver positive long-run buy-and-hold returns, the median token loses roughly 30\% of its value over one to twenty-four months, with large issuances more often overpriced and underperforming. \cite{hoang-2021} use high-frequency data to show that no major stablecoin is absolutely stable, that return variance is non-zero, and that volatility is state dependent. Related studies document that stability varies meaningfully across designs and over time \citep{grobys-2021,jarno-2021}, and that some stablecoins can act as conditional diversifiers or safe havens relative to Bitcoin during stress \citep{baur-hoang-2020,wang-2020}. These findings imply that operational choices inside the issuer, such as cash buffers, rebalancing cadence, and fee schedules, can shape realized peg behavior observed in secondary markets.

Price formation around the peg can be modeled by combining primary arbitrage with secondary-market order-flow dynamics. \cite{pernice-2021} formalizes a fully collateralized environment in which arbitrage against the issuer and trend-following demand interact with deviations from fundamental value. Empirically, peg deviations and trend variables predict short-horizon returns when extreme observations are retained, while explanatory power shrinks once large dislocations are trimmed \citep{pernice-2021}. This pattern is consistent with frictions and regime-like behavior, and it motivates state-dependent fee and inventory policies that anticipate clustered dislocations rather than treating all deviations as transitory noise.

A second body of evidence connects visible stablecoin flows to activity in major crypto assets. \cite{ante-2021} study 1{,}587 on-chain stablecoin transfers of at least one million U.S.\ dollars and find significant abnormal increases in Bitcoin trading volume around these transfers. High-frequency work on Tether corroborates flow-through to returns and microstructure \citep{grobys-huynh-2021}, and pandemic-period analyses report regime-dependent behavior consistent with liquidity demand and risk rebalancing \citep{jeger-2020}. These results support modeling mint and redemption dynamics as first-order drivers of liquidity demand that issuers must be prepared to meet.

The causal direction between issuance and prices remains debated. \cite{griffin-2018} present evidence that Tether flows arrive after downturns and predict positive subsequent Bitcoin returns; \cite{wei-2018} also report issuance-related effects on Bitcoin. In contrast, \cite{kristoufek-2021} finds that issuance largely reacts to crypto price moves, with weak net spillovers from stablecoins to major crypto assets once contemporaneous correlations are controlled. Both mechanisms can surface in different market states, which is consistent with a policy problem in which the issuer prepares for rare but severe clustering of outflows without sacrificing unnecessary carry in normal conditions.

Macro-financial conditions further shape stablecoin behavior and the reserve allocation trade-off. \cite{nguyen-2022} document that increases in interbank rates compress prices and volatility of leading stablecoins while raising both price levels and volatility of traditional cryptocurrencies; trading activity increases across both groups. Beyond market microstructure, general-equilibrium perspectives suggest that widespread adoption can carry macroeconomic implications for money and credit \citep{bojaj-2022}. For fiat-backed issuers, rate regimes move the opportunity cost of holding cash relative to short-duration government bills. When redemption risk is elevated, the opportunity cost of liquidity is not static; it shifts with the interest-rate environment in ways that directly influence the optimal cash share.

Connectedness across stablecoins introduces sector-level channels of stress transmission and substitution. \cite{thanh-2022} propose complementary stability metrics and show that large USD-backed coins transmit stability conditions to smaller peers. Studies of safe-haven and hedging properties also reveal state dependence across assets and sectors \citep{wang-2020}, with COVID-era work on gold-backed tokens highlighting additional substitution and spillover channels \citep{jalan-2021,aloui-2020,yousaf-2021}. These insights justify scenario designs in which a few dominant issuers shape the overall flow environment faced by the sector.

Stress episodes in non-fiat-backed designs provide quantitative benchmarks for operational frictions, capacity constraints, and governance responses. \cite{kjaer-2021} analyze MakerDAO around March 2020 and document how congestion concentrated liquidations into narrow windows and reduced auction effectiveness. The Terra--Luna collapse illustrates a distinct failure mode. \cite{uhlig-2022} develops a continuous-time account linking burn dynamics and suspension beliefs; \cite{briola-2022} assemble a high-frequency timeline of the unwind; and recent modeling of algorithmic stablecoin volatility complements these post-mortems \citep{zhao-2021}. While our scope is fiat-backed issuers operating with high-quality reserves, these episodes underscore the need for explicit, stress-aware liquidity policy at intra-day horizons and motivate the inclusion of settlement windows, queueing effects, and friction costs in our objective.

An analysis in \cite{catalini2021setting} documents substantial heterogeneity in reserve practices among major stablecoin issuers in 2021: Tether managed roughly USD 63 billion across a diversified set of instruments with heterogeneous liquidity, whereas USDC held about USD 32 billion entirely in cash; Diem, by contrast, recommended holding ~10\% in cash and ~90\% in short-term government securities. This work extends both the academic literature and current treasury practice by developing a control-theoretic reserve-management framework that adapts the portfolio in real time to observed flows and yields while respecting current legislation and asset-quality constraints. Our policy replaces static, rule-based heuristics with an implementable, settlement-aware, closed-loop design that (i) delivers transparent, window-based rebalancing rules with explicit execution frictions and (ii) maximizes expected carry subject to quantified de-peg penalties and hard cash-availability constraints. 

The remainder of the paper is structured as follows. Section~\ref{sec:system-framework} (Reserve Management Framework) introduces the reserve-eligibility set, the regulatory liquidity hierarchy, the dynamic reserve model with peg dynamics, and the state and control variables used throughout. Within Section~\ref{sec:system-framework}, Subsection~\ref{subsec:dynamic_reserve_model} develops the dynamic reserve model that links flows, liquidity, and peg deviations. Section~\ref{sec:methods} develops the solution approach, including stochastic model predictive control, the Pontryagin conditions that generate feedback policies, and a settlement-window implementation; the sampled-data control used in practice is detailed in Subsection~\ref{subsec:settlement}. Section~\ref{sec:numerical_results} specifies scenarios and parameters, compares the controller with liquidity- and yield-focused benchmarks, and reports sensitivity analyses. Section~\ref{sec:conclusion} concludes.

\section{Reserve Management Framework}
\label{sec:system-framework}

\subsection{Regulatory Requirements}
\label{subsec:regulatory_requirements}

Regulatory authorities worldwide enforce stringent requirements concerning the asset classes permissible as collateral for stablecoin issuers. In this paper, we adopt the U.S. regulatory framework as our primary reference point for modeling purposes. This approach is motivated by several factors: the maturity and transparency of U.S. financial regulation, its substantial global influence, and the availability of comprehensive legislative proposals such as the \textsc{Genius} Act. While our analysis focuses on U.S. regulations, the resulting models readily generalize to other jurisdictions, as most major economies have converged on similar regulatory principles for stablecoins—particularly the requirements for full collateralization through high-quality liquid assets and robust redemption guarantees. Under the U.S. regulatory framework, Title VII, Section 604(b)(1) of the Guiding and Establishing National Innovation for U.S. Stablecoins Act of 2025 (the "GENIUS Act") delineates six categories of high-quality, liquid assets eligible to serve as collateral for payment stablecoins on a strict 1:1 basis \cite{GENIUSAct2025}:
\begin{enumerate}
    
\item \textbf{U.S.\ currency and balances held directly at a Federal Reserve Bank}

\item \textbf{Demand deposits at FDIC-insured depository institutions}; that is, checking or transaction accounts held at commercial banks, insured by the U.S.\ government against default.

\item \textbf{U.S.\ Treasury securities with residual maturity~$\leq90$\,days}; that is, short-term federal debt instruments.

\item \textbf{Overnight reverse-repo agreements backed by Treasuries, cleared tri-party or with a highly credit-worthy counterparty}; i.e., short-term lending arrangements where the issuer temporarily exchanges cash for Treasury collateral, typically settled via a clearing bank.

\item \textbf{Government-only money-market fund shares invested solely in assets~(1)–(4)}; that is, pooled investment vehicles offering liquidity and diversification while limiting holdings to ultra-safe, short-term instruments.

\item \textbf{Tokenised forms of any asset in~(1)–(5), excluding repo contracts}; that is, cryptographic representations of reserve assets issued on a blockchain, provided the underlying remains fully redeemable and auditable.

\end{enumerate}

These six categories establish a natural liquidity hierarchy that reflects different settlement windows and operational characteristics. At the apex of this hierarchy are cash balances and Federal Reserve master accounts, which provide instantaneous settlement capability but generate no yield, creating an inflation drag on reserves while offering complete elimination of market risk. Shortterm U.S. Treasury securities with maturities of 90 days or less generally settle within standard banking hours, providing substantially higher yields that currently approximate 5.4\% p.a. while maintaining deep market liquidity and qualifying as Level-1 High Quality Liquid Assets under Basel III regulations. However, these instruments introduce duration risk and potential fire-sale losses during periods of intense redemption pressure.
The remaining asset classes exhibit varying combinations of liquidity constraints and settlement windows that create a graduated spectrum of operational flexibility. FDIC-insured demand deposits typically process during Fedwire operating windows, while government money market fund shares offer same-day liquidity under normal market conditions but may impose restrictions during periods of market stress. Tokenized representations of these assets provide 24/7 trading capability through blockchain networks, though final settlement of the wrapped asset remains constrained by the operational windows of the native market.

Throughout, we model cash as operationally
available on demand. In real treasury operations, "cash" often means balances
at commercial banks or custodians and can therefore carry counterparty and access risk. This is
consistent with episodes where a stablecoin’s secondary-market price deviated
from par due to perceived custodial risk even when the issuer remained solvent.
Operationally, issuers mitigate this by (i) diversifying transaction-account
balances across multiple depository institutions, (ii) preferring central-bank
accounts where available, (iii) sweeping excess cash into Treasury-only vehicles
with same-day liquidity, and (iv) maintaining committed liquidity lines for
intraday settlement shocks.

The fundamental trade-off between cash and Treasury securities illustrates the central tension in stablecoin reserve management. Cash holdings maximize liquidity and eliminate market risk but impose opportunity costs through foregone interest income, particularly pronounced in rising rate environments. Treasury securities offer attractive risk-adjusted returns and maintain high liquidity under normal conditions but introduce duration risk and volatility during market dislocations. The other permitted asset classes present similar risk-return considerations, each with specific additional complexities such as smart contract and legal title risks introduced by tokenized wrappers.

Throughout our analysis, we adopt benchmark parameters that reflect current market conditions, with the risk-free cash rate set at zero as vault cash earns no interest, while three-month Treasury securities yield approximately 5.4\% p.a. based on prevailing yield curves.\footnote{Indicative annualized yields as of 2Aug2025.}

\subsection{A Dynamic Reserve Model}
\label{subsec:dynamic_reserve_model}

\subsubsection{Demand And Peg Stability}
\label{subsubsec:demand_peg_stability}

The peg deviation $\Delta P(t) \in \mathbb{R}$ measures the stablecoin's price displacement from its \$1 target, where $\Delta P(t) > 0$ indicates the coin trades below par. We model the stablecoin's mint and redemption dynamics as self-exciting processes that capture the empirical tendency for large transactions to cluster in time, while the peg deviation evolves deterministically based on liquidity shortfalls. Let $N_{R}(t)$ and $N_{M}(t)$ denote the cumulative redemption and mint counts up to time $t$, each following a univariate Hawkes process with stochastic intensities

\begin{align}
  \lambda_{R}(t) &= \lambda_{0,R}
    + \int_{0}^{t} \kappa_R e^{-\theta_R(t-s)} \,\mathrm{d}N_{R}(s)
    + \int_{0}^{t} \zeta\, e^{-\theta_R(t-s)} \Delta P(s)^2\,\mathrm{d}s, \\
  \lambda_{M}(t) &= \lambda_{0,M}
    + \int_{0}^{t} \kappa_M e^{-\theta_M(t-s)} \,\mathrm{d}N_{M}(s).
\end{align}
where $\lambda_{0,R}, \lambda_{0,M} \geq 0$ are baseline intensities, $\kappa_R, \kappa_M > 0$ capture self-excitation (each event increases the likelihood of subsequent events), and $\theta_R, \theta_M > 0$ govern the exponential decay of this excitation. The redemption intensity includes an additional feedback term driven by $\Delta P(s)^2$, reflecting how peg instability accelerates redemptions. This mechanism endogenously models bank-run dynamics characterized by self-reinforcing distrust in the currency's stability.

Let $M(t)$ and $R(t)$ denote cumulative dollar volumes of mints and redemptions, with instantaneous net redemption flow $Q(t) = \dot{R}(t) - \dot{M}(t)$. These dollar flows arise from the counting processes via marked point process representations:
\begin{equation}
\label{eq:marks-executed}
\mathrm{d}M(t) = \int_{\mathbb{R}_+} z \, N_M(\mathrm{d}t, \mathrm{d}z),
\qquad
\mathrm{d}R(t) = \int_{\mathbb{R}_+} \min\bigl(z,\, S_{\mathrm{out}}(t^-)\bigr)\, N_R(\mathrm{d}t, \mathrm{d}z),
\end{equation}
where marks $z$ represent individual transaction sizes with finite first moments. 
To rule out mechanically infeasible paths in which cumulative executed redemptions exceed the outstanding supply, we truncate each redemption mark at the remaining supply. The outstanding supply evolves as
\begin{equation}
\label{eq:Sout-def}
S_{\mathrm{out}}(t)=S_{\mathrm{out}}(0)+M(t)-R(t), \qquad S_{\mathrm{out}}(0)>0,
\end{equation}
where the left-limit $S_{\mathrm{out}}(t^-)$ in \eqref{eq:marks-executed} ensures that executed redemptions satisfy $R(t)\le S_{\mathrm{out}}(0)+M(t)$ and hence $S_{\mathrm{out}}(t)\ge 0$ for all $t$.

The conditional expected net cash outflow rate is
\begin{equation} \label{eq:Qhat-def}
  \hat{Q}(t) := \mathbb{E}\!\left[Q(t)\mid\mathcal{F}_t\right]
             = \bar z_R\,\lambda_R(t) - \bar z_M\,\lambda_M(t),
\end{equation}
where $\bar z_R$ and $\bar z_M$ denote the mean redemption and mint sizes,
respectively.  Over a settlement horizon $\tau_{\mathrm{set}}>0$, the
quantity $\tau_{\mathrm{set}}\,[\hat{Q}(t)]_+$ represents the expected
cumulative cash drain that remains after netting mints against
redemptions.  Comparing this drain with the immediately available cash
reserve $R_{\mathrm{liq}}(t)$ yields a normalized liquidity shortfall
\begin{equation} \label{eq:DeltaPfair-def}
  \Delta P_{\mathrm{fair}}(t)
  := \mathbf{1}_{\{S_{\mathrm{out}}(t)>0\}}
     \left[
       \frac{\tau_{\mathrm{set}}\,[\hat{Q}(t)]_+ - R_{\mathrm{liq}}(t)}
            {S_{\mathrm{out}}(t)}
     \right]_{0}^{1},
\end{equation}
where $[x]_+ := \max\{x,0\}$ and $[x]_{0}^{1}:=\min\{1,\max\{x,0\}\}$.
The quantity $\Delta P_{\mathrm{fair}}(t)$ measures the fraction of
outstanding supply whose redemption cannot be settled from current liquid
reserves within the relevant horizon. We interpret $\Delta P_{\mathrm{fair}}(t)$ as
the fair-value peg deviation that would prevail if the secondary-market
price instantaneously reflected this settlement-liquidity gap.

The actual peg deviation then evolves deterministically through
first-order mean reversion toward $\Delta P_{\mathrm{fair}}(t)$ at rate
$\eta > 0$, supplemented by a corrective term proportional to the
mint/burn fee spread $\delta(t)$:
\begin{equation} \label{eq:DeltaP}
  \dot{\Delta P}(t)
  = \eta\!\Big(\Delta P_{\mathrm{fair}}(t) - \Delta P(t)\Big)
    - \gamma\,\delta(t).
\end{equation}
The parameter $\eta > 0$ governs this sensitivity, while $\gamma > 0$ determines how effectively fees can counteract peg deviations. This deterministic relationship makes peg stability an operational challenge: maintaining $\Delta P(t) \approx 0$ requires either sufficient liquidity buffers or appropriate fee policies to offset redemption pressure.

\subsubsection{Reserve-Balance State Model}
\label{subsubsec:reverse_balance_state_model}

 In the remaining text, we make the simplifying assumption that the coin issuer employs only vault cash (categories 1-2 from subsection \ref{subsec:regulatory_requirements}) and short-dated Treasury securities (category 3) as reserve assets. This simplification is due to the fact that the remaining permitted asset types primarily represent overnight carries or swaps of these two base assets, serving as operational vehicles rather than fundamentally distinct reserve instruments. The stablecoin issuer's operational challenge involves dynamically allocating reserves between immediately available cash and yield-bearing Treasury securities while maintaining peg stability. We formalize this as a control problem with state vector
\begin{equation}
x(t) = \begin{bmatrix}
R_{\mathrm{liq}}(t) \\[2pt]
R_{\mathrm{bill}}(t) \\[2pt]
\Delta P(t) \\[2pt]
\lambda_R(t) \\[2pt]
\lambda_M(t)
\end{bmatrix},
\end{equation}
where $R_{\mathrm{liq}}(t)$ represents immediately available cash reserves, $R_{\mathrm{bill}}(t)$ denotes the face value of short-dated Treasury securities, $\Delta P(t)$, and $\lambda_R(t)$, $\lambda_M(t)$ are the Hawkes intensities governing redemption and mint clustering. The controller manages this system through
\begin{equation}
u(t) = \begin{bmatrix} \omega(t) \\ \delta(t) \end{bmatrix}
\;\in\;
[-\omega_{\max},\, \omega_{\max}] \times [-\delta_{\max},\, \delta_{\max}],
\label{eq:control-box}
\end{equation}
where \( \omega(t) \) represents the reallocation flow between cash and bills (positive for cash-to-bills, negative for bills-to-cash), and \( \delta(t) \) is the mint/burn fee spread.

The reserve balances evolve according to
\begin{align}
\mathrm{d}R_{\mathrm{liq}}(t) &= r_{\mathrm{cash}} R_{\mathrm{liq}}(t) \mathrm{d}t - Q(t) \mathrm{d}t - \omega(t) \mathrm{d}t, \label{eq:Rliq}\\[4pt]
\mathrm{d}R_{\mathrm{bill}}(t) &= r_{\mathrm{bill}}(t) R_{\mathrm{bill}}(t) \mathrm{d}t + \omega(t) \mathrm{d}t. \label{eq:Rbill}
\end{align}
For tractability within control horizons $[t_k, t_{k+1}]$, we treat the T-bill rate as piecewise constant: $r_{\mathrm{bill}}(t) \equiv r_{\mathrm{bill}}^{(k)}$ for $t \in [t_k, t_{k+1}]$. Our baseline assumes $r_{\mathrm{cash}} = 0$, reflecting that vault cash and Federal Reserve master account balances earn no interest, though the framework accommodates $r_{\mathrm{cash}} \neq 0$ without structural changes.

We impose the constraint $R_{\mathrm{bill}}(t)\ge 0$ for all $t$, since the bill bucket cannot be negative. In contrast, the liquid cash position can be temporarily negative in practice due to intraday settlement mismatches and/or the use of committed liquidity lines. Accordingly, we allow $R_{\mathrm{liq}}(t)\in\mathbb{R}$. Combined with the bounded reallocation rate $|\omega(t)| \leq \omega_{\max}$, the constraints endogenize the fundamental liquidity-yield trade-off: higher Treasury allocations increase returns but reduce the buffer against redemption surges, potentially triggering peg instability through the mechanism in equation \eqref{eq:DeltaP}. The reserve base satisfies $R_{\mathrm{liq}}(t) + R_{\mathrm{bill}}(t) \ge S_{\mathrm{out}}(t)$ for all $t\ge 0$ whenever the inequality holds at $t=0$, since interest accrual on Treasury holdings generates an excess $r_{\mathrm{bill}}\,R_{\mathrm{bill}}\ge 0$, which means the issuer is overcollateralized by the cumulative interest earned.

\subsubsection{Optimal Control Problem}
\label{subsubsec:optimal_control_problem}

The stablecoin issuer seeks to maximize yield while maintaining peg stability and operational efficiency. We formalize this through a cost functional that captures the competing objectives of interest generation, peg quality, user experience, and operational frictions. The instantaneous cost is
\begin{equation}
\ell\bigl(x(t),u(t)\bigr)
=
c_{\mathrm{peg}}\,\Delta P(t)^{2}
\;+\;
c_{\mathrm{fee}}\,\delta(t)^{2}
\;+\;
\lambda_{\omega}\,|\omega(t)|
\;+\;
\frac{1}{2}\rho_{\omega}\,\omega(t)^{2}
\;-\;
r_{\mathrm{bill}}^{(k)}\,R_{\mathrm{bill}}(t),
\label{eq:stage-cost}
\end{equation}
where the terms represent: (i) peg stability penalty $c_{\mathrm{peg}} \Delta P(t)^2$ monetizing reputational and convertibility risks from price deviations; (ii) fee friction $c_{\mathrm{fee}} \delta(t)^2$ discouraging excessive mint/burn spreads that degrade user experience; (iii) reallocation cost $\lambda_{\omega}\,|\omega(t)|
\;+\;
\frac{1}{2}\rho_{\omega}\,\omega(t)^{2}$ capturing operational frictions and market impact from portfolio adjustments; and (iv) interest benefit $-r_{\mathrm{bill}}^{(k)} R_{\mathrm{bill}}(t)$ representing carry from Treasury holdings. The quadratic penalties ensure convexity while the linear yield term creates the fundamental tension between liquidity provision and return generation.

The infinite-horizon stochastic optimal control problem is to select an admissible control process $u(\cdot) = (\omega(\cdot), \delta(\cdot))$ minimizing the expected discounted cost:
\begin{equation}
J_{\infty}(x(0)) = \min_{u(\cdot) \in \mathcal{U}} \mathbb{E}\left[ \int_0^{\infty} e^{-\rho t} \ell(x(t), u(t)) \mathrm{d}t \right],
\label{eq:OCP-infinite}
\end{equation}
subject to the system dynamics
\begin{align}
\mathrm{d}R_{\mathrm{liq}}(t) &= r_{\mathrm{cash}} R_{\mathrm{liq}}(t) \mathrm{d}t - Q(t) \mathrm{d}t - \omega(t) \mathrm{d}t, \label{eq:dyn-Rliq}\\
\mathrm{d}R_{\mathrm{bill}}(t) &= r_{\mathrm{bill}}^{(k)} R_{\mathrm{bill}}(t) \mathrm{d}t + \omega(t) \mathrm{d}t, \label{eq:dyn-Rbill}\\
\dot{\Delta P}(t)
&=
\eta\big(\Delta P_{\mathrm{fair}}(t)-\Delta P(t)\big)-\gamma\,\delta(t),
\qquad \Delta P_{\mathrm{fair}}(t)\ \text{as in \eqref{eq:DeltaPfair-def}}.\\
\mathrm{d}\lambda_R(t) &= -\theta_R(\lambda_R(t) - \lambda_{0,R}) \mathrm{d}t + \kappa_R \mathrm{d}N_R(t) + \zeta \Delta P(t)^2 \mathrm{d}t, \label{eq:dyn-lamR}\\
\mathrm{d}\lambda_M(t) &= -\theta_M(\lambda_M(t) - \lambda_{0,M}) \mathrm{d}t + \kappa_M \mathrm{d}N_M(t), \label{eq:dyn-lamM}
\end{align}
where the net redemption flow $Q(t) = \dot{R}(t) - \dot{M}(t)$ arises from marked Hawkes processes with intensities $\lambda_R(t)$, $\lambda_M(t)$ and trade-size marks of finite first moment.

Admissible controls must satisfy the operational constraints \eqref{eq:control-box}, be progressively measurable with respect to the filtration $\{\mathcal{F}_t\}$, and maintain state feasibility with $R_{\mathrm{bill}}(t)\ge 0$.

The infinite-horizon optimization problem faces several well-posedness challenges: convergence of the cost integral despite potentially unbounded Treasury growth, the non-smooth clipping in the fair-value mapping
$\Delta P_{\mathrm{fair}}(t)$ defined in~\eqref{eq:DeltaPfair-def} appearing in the peg dynamics, and potential explosive paths arising from the peg-redemption feedback loop. We address these challenges through the following conditions: (i) a positive discount rate $\rho > 0$ with bounded interest rates; (ii) subcritical Hawkes processes satisfying $\kappa_R/\theta_R < 1$ and $\kappa_M/\theta_M < 1$, which guarantee stationary intensities and finite first moments; (iii) bounded controls that maintain nonnegative reserves; and (iv) the balance sheet constraint that prevents unbounded carry extraction by tying reserves to outstanding supply.

A crucial observation is that the problem admits a finite-value admissible policy, which serves as a witness for well-posedness. Since every outstanding coin is backed one-to-one, we can construct a peg-neutral policy with finite discounted cost. Consider the all-cash policy defined by
\begin{equation}
u^{\mathrm{cash}}(t) = (\omega(t), \delta(t)) \equiv (0, 0), \quad R_{\mathrm{bill}}(t) \equiv 0, \quad R_{\mathrm{liq}}(t) = S_{\mathrm{out}}(t),
\end{equation}
where the issuer holds the full reserve in immediately available cash without trading between buckets or charging fees.

Under this all-cash policy, immediately available cash equals outstanding supply at all times, so
\[
\Delta P_{\mathrm{fair}}(t)
=
\left[1-\frac{R_{\mathrm{liq}}(t)}{S_{\mathrm{out}}(t)}\right]_+
=
\left[1-\frac{S_{\mathrm{out}}(t)}{S_{\mathrm{out}}(t)}\right]_+
=0.
\]
The peg dynamics therefore reduce to
\[
\dot{\Delta P}(t) = -\eta\,\Delta P(t) - \gamma\,\delta(t).
\]
With $\delta(t)\equiv 0$ and the natural initialization $\Delta P(0)=0$, we obtain $\Delta P(t)\equiv 0$ for all $t\ge 0$.

For this admissible trajectory, the running cost simplifies to
\begin{equation}
\ell\bigl(x(t), u^{\mathrm{cash}}(t)\bigr) = c_{\mathrm{peg}} \cdot 0 + c_{\mathrm{fee}} \cdot 0 + \lambda_{\omega} \cdot 0 + \tfrac{1}{2}\rho_{\omega} \cdot 0 - r_{\mathrm{bill}}(t) \cdot 0 = 0,
\end{equation}
and consequently the discounted infinite-horizon objective equals
\begin{equation}
J_\infty\bigl(x(0); u^{\mathrm{cash}}\bigr) = \int_0^\infty e^{-\rho t} \cdot 0 \, dt = 0 < \infty.
\end{equation}

Therefore, the value function is finite at $x(0)$ and the optimization problem is proper: the infimum of the discounted cost is not $+\infty$ and a minimizing sequence exists. This witness policy demonstrates why the peg-redemption feedback cannot cause explosive costs in our formulation - there exists an admissible control that nullifies the fair-value deviation term $\Delta P_{\mathrm{fair}}$ in the peg law, keeping $\Delta P$ identically zero with no associated cost.

Together with the existence of the admissible peg-neutral policy demonstrated above, these conditions ensure that the discounted infinite-horizon problem is well-posed: the value function remains finite for the initial state, and the state-control constraints admit at least one trajectory with bounded cost. This establishes an economically meaningful baseline of perfect convertibility with zero peg error, against which all other policies can be evaluated within the same discounted framework. 

\subsubsection{Assumptions and robustness discussion}
\label{subsubsec:assumptions_robustness}
The model is deliberately stylized to yield an auditable, settlement-aware feedback rule rather than a fully structural description of stablecoin secondary-market pricing. This subsection discusses which assumptions are primarily adopted for tractability, how they can affect quantitative outputs, and how they may be relaxed while preserving the qualitative policy logic (threshold rebalancing and monotone stress response):

\begin{enumerate}
\item \textit{Flow model (Hawkes) and calibration.}
We use exponential Hawkes processes because they reproduce clustered order flow
with a Markovian intensity state, which is what enables the moment-closure
forecast and the MPC implementation. Empirically, the Hawkes parameters
$(\lambda_{0,\cdot},\kappa_{\cdot},\theta_{\cdot})$ can in principle be estimated directly from on-chain mint/burn timestamps and updated over rolling windows. These estimates can be fed into the filtered intensity updates
\eqref{eq:jump}-\eqref{eq:decay}.
We emphasize that the degree of self-excitation (branching ratios
$\kappa/\theta$) is quantitatively important because it affects the speed and
persistence of clustered redemptions and therefore the optimal liquidity buffer.
Our scenario design partially addresses this sensitivity by varying both baseline
intensities and moving the redemption branching ratio close to
criticality, which should be interpreted as a robustness check over
plausible flow regimes rather than a point-calibrated empirical claim.

\item \textit{Objective specification and sensitivity.}
The stage cost \eqref{eq:stage-cost} is intended as a reduced-form aggregation
of (i) peg-quality and reputational costs, (ii) user-friction costs from fees,
(iii) trading frictions/market impact, and (iv) bill carry. The calibration in
Section~\ref{subsec:parameter_estimation} should be interpreted as illustrative scaling chosen to match operational tolerances.
While alternative penalty choices would shift numerical thresholds and cash targets, the
qualitative control form is robust to these
variations. The framework is designed to be re-calibrated to each issuer's
mandate rather than used with universal weights.

\item \textit{Interpretation and limitation (confidence-driven de-pegs).}
Equation~\eqref{eq:DeltaP} models an operational liquidity-driven
component of de-pegging: $\Delta P$ increases when redemptions exceed
immediately available settlement liquidity within the relevant window. Historical
episodes suggest that secondary-market prices can also deviate from par due to
anticipatory fear even when
redemptions remain available. Our formulation intentionally isolates the
liquidity/settlement channel because it is the part most directly controllable
by treasury allocation and execution timing.
\end{enumerate}

\section{Analytical and Computational Methods}
\label{sec:methods}

\subsection{Stochastic Model Predictive Control}
\label{subsec:smpc-framework}

The infinite-horizon optimal control problem formulated above presents significant computational challenges for real-time implementation. This motivates our adoption of Model Predictive Control (MPC), a rolling optimization strategy that approximates the infinite-horizon solution through repeated finite-period problems. Model Predictive Control implements the following strategy: at each decision time $t_k$, we solve a finite-horizon optimal control problem using the current state and forecasts of future uncertainty. Only the initial portion of the computed optimal trajectory is implemented before re-solving at the next decision point $t_{k+1}$ with updated information. This approach naturally handles constraints while incorporating predictions of future system behavior. For stablecoin reserve management, MPC offers several key advantages: it requires no prior training data, making it immediately applicable without extensive historical datasets; it produces transparent and auditable policies through explicit optimization of well-defined objectives, ensuring regulatory compliance; it naturally incorporates forward-looking capability to anticipate redemption clusters and settlement windows; and it effectively handles high-dimensional continuous state spaces and hybrid jump dynamics through direct optimization.

At each roll time $t_k$, we solve the finite-horizon stochastic optimal control problem
\begin{equation}
\min_{u(\cdot)\in\mathcal{U}} J_k = \mathbb{E}\left[ \int_{t_k}^{t_k+T} e^{-\rho(t-t_k)} \ell(x(t), u(t)) \mathrm{d}t \;\Big|\; \mathcal{F}_{t_k} \right]
\label{eq:smpc-ocp}
\end{equation}
subject to the system dynamics and admissibility constraints. The filtration $\mathcal{F}_{t_k}$ represents all information available at time $t_k$, including the complete history of mint and redemption events, current reserve levels, and realized peg deviations. As a measure-theoretic object, $\mathcal{F}_{t_k}$ is the $\sigma$-algebra generated by all observable processes up to time $t_k$, ensuring that our control decisions are non-anticipating—they depend only on past and present information, not future realizations.

Our stochastic MPC framework reduces the stochastic optimization problem to a deterministic equivalent through moment-based approximation of the underlying stochastic processes. Rather than propagating full probability distributions or sampling scenarios, we replace all random quantities by their conditional expectations given $\mathcal{F}_{t_k}$. This transformation occurs at two levels. First, between MPC rolls, we maintain filtered estimates of the Hawkes intensities using the recursive updates
\begin{align}
\lambda_X(t^+) &= \lambda_X(t^-) + \kappa_X &&\text{(at event)}, \label{eq:jump} \\
\lambda_X(t+\Delta)^- &= \lambda_{0,X} + \big(\lambda_X(t^+) - \lambda_{0,X}\big)e^{-\theta_X \Delta} 
&&\text{(between events)}. \label{eq:decay}
\end{align}
for $X \in \{R,M\}$. These updates are deterministic given the observed event stream, providing initial conditions $\lambda_R(t_k^+)$ and $\lambda_M(t_k^+)$ for the planning phase.

Within the optimization horizon $[t_k, t_k+T]$, we must propagate conditional mean intensities to compute the expected redemption flow $\hat{Q}(t)$ at each instant. Simply using initial intensities would ignore the predictable evolution of redemption pressure over the planning horizon. The conditional mean intensities evolve according to
\begin{align}
\dot{\hat{\lambda}}_R(t) &= -\theta_R(\hat{\lambda}_R(t) - \lambda_{0,R}) + \kappa_R \hat{\lambda}_R(t) + \zeta \Delta P(t)^2, \quad \hat{\lambda}_R(t_k) = \lambda_R(t_k^+), \\
\dot{\hat{\lambda}}_M(t) &= -\theta_M(\hat{\lambda}_M(t) - \lambda_{0,M}) + \kappa_M \hat{\lambda}_M(t), \quad \hat{\lambda}_M(t_k) = \lambda_M(t_k^+),
\end{align}
where the self-excitation terms $\kappa_X \hat{\lambda}_X(t)$ arise from taking expectations of future jump contributions. These ODEs exploit the Markovian property of exponential Hawkes processes: given current intensities, the expected future evolution depends only on these values, not the full history. The expected outstanding supply within the MPC horizon evolves as
\begin{equation}
\label{eq:Sout-hat}
\dot{\hat{S}}_{\mathrm{out}}(t)=\bar z_M \hat{\lambda}_M(t)-\bar z_R \hat{\lambda}_R(t)=-\hat{Q}(t),
\qquad \hat{S}_{\mathrm{out}}(t_k)=S_{\mathrm{out}}(t_k).
\end{equation} Substituting these deterministic forecasts into the system dynamics yields
\begin{align}
\dot{R}_{\mathrm{liq}}(t) &= r_{\mathrm{cash}} R_{\mathrm{liq}}(t) - \hat{Q}(t) - \omega(t), \\
\dot{R}_{\mathrm{bill}}(t) &= r_{\mathrm{bill}}^{(k)} R_{\mathrm{bill}}(t) + \omega(t), \\
\dot{\Delta P}(t)
&=
\eta\big(\hat{\Delta P}_{\mathrm{fair}}(t)-\Delta P(t)\big)-\gamma\,\delta(t),\\
\hat{\Delta P}_{\mathrm{fair}}(t)
&:=
\mathbf{1}_{\{\hat{S}_{\mathrm{out}}(t)>0\}}
\left[
\frac{\tau_{\mathrm{set}} \,[\hat{Q}(t)]_+ - R_{\mathrm{liq}}(t)}{\hat{S}_{\mathrm{out}}(t)}
\right]_{0}^{1}.
\end{align}
where $\hat{S}_{\mathrm{out}}(t)$ evolves according to the expected mint and redemption flows. The computed trajectory $u^*(t)$ for $t \in [t_k, t_k+T]$ represents the optimal response to expected future conditions, with only the initial segment $u^*(t)$ for $t \in [t_k, t_{k+1})$ implemented before re-solving with updated state information.

\subsection{Optimal Control Characterization via Pontryagin's Maximum Principle}

Within each MPC roll, we solve the deterministic surrogate problem obtained by replacing stochastic flows with their conditional expectations. The state vector $x = (R_{\mathrm{liq}}, R_{\mathrm{bill}}, \Delta P)^{\top}$ evolves according to the moment-closed dynamics from Section \ref{subsec:smpc-framework}, while the Hawkes intensities $\hat{\lambda}_R, \hat{\lambda}_M$ enter as exogenous trajectories. 

We seek an optimal control trajectory \(u(\cdot)\) from a suitable function space, such as the Banach space of square-integrable functions \(L^{2}\!\left([0,T], \mathbb{R}^{m}\right)\), rather than a finite-dimensional vector. This makes the optimization problem infinite-dimensional. To solve the optimal control problem, Pontryagin's Maximum Principle (PMP) transforms the dynamic optimization into a system of differential equations. The optimal control at each instant must minimize the current-value Hamiltonian
\begin{equation}
H = \ell(x,u) + p^{\top} f(x,u,t),
\end{equation}
which captures the total instantaneous cost: the direct running cost $\ell(x,u)$ plus the marginal effect of the control on future costs through its impact on state evolution $f(x,u,t)$. The optimal control thus balances immediate costs against future consequences. The costates $p = (p_{\mathrm{liq}}, p_{\mathrm{bill}}, p_{\Delta P})^{\top}$ can consequently be interpreted as shadow prices or marginal values of the state variables, measuring how a unit change in each state affects the total objective value going forward.

Following Pontryagin's Maximum Principle for infinite-horizon problems, the optimal trajectory must satisfy the following canonical equations:
\begin{align}
\dot{x}^* &= f(x^*, u^*, t), \quad x^*(0) = x_0, \\
\dot{p} &= \rho p - \nabla_x H(x^*, u^*, p, t) \label{eq:costate}, \\
u^*(t) &\in \arg\min_{u \in \mathcal{U}} H(x^*(t), u, p(t), t),
\end{align}
where the first term in~\eqref{eq:costate} reflects the accumulation of shadow costs at the discount rate, while the second term accounts for their depletion as costs are realized over the infinitesimal interval~$[t, t + dt]$. The transversality condition $\lim_{t \to \infty} e^{-\rho t} p^{\top}(t) x^*(t) = 0$ ensures that the discounted value of the terminal state vanishes, preventing unbounded growth of the state-costate product. Since our cost functional is convex in controls and the reserve dynamics are affine, these necessary conditions from PMP are also sufficient for optimality, guaranteeing that any solution to this system is globally optimal.

The costate dynamics yield explicit expressions for the shadow prices. For the current-value Hamiltonian
\begin{align}
H &= c_{\mathrm{peg}} \Delta P^2 + c_{\mathrm{fee}} \delta^2 + \lambda_{\omega}|\omega| + \tfrac{1}{2}\rho_{\omega} \omega^2 - r_{\mathrm{bill}} R_{\mathrm{bill}} \nonumber\\
&\quad + p_{\mathrm{liq}}(r_{\mathrm{cash}} R_{\mathrm{liq}} - \hat{Q} - \omega) + p_{\mathrm{bill}}(r_{\mathrm{bill}} R_{\mathrm{bill}} + \omega) \nonumber\\
&\quad + p_{\Delta P}\!\left(
\eta\!\left(
\mathbf{1}_{\{\hat{S}_{\mathrm{out}}>0\}}
\!\left[
\frac{\tau_{\mathrm{set}}[\hat{Q}]_+ - R_{\mathrm{liq}}}
     {\hat{S}_{\mathrm{out}}}
\right]_0^1
\!- \Delta P
\right)
-\gamma\,\delta
\right),
\intertext{the costates satisfy (almost everywhere at points where $R_{\mathrm{liq}} \neq \hat{S}_{\mathrm{out}}$)}
\dot{p}_{\mathrm{liq}}
&=
\rho p_{\mathrm{liq}}-\frac{\partial H}{\partial R_{\mathrm{liq}}}
=
(\rho-r_{\mathrm{cash}})p_{\mathrm{liq}}
+
p_{\Delta P}\,\eta\,
\frac{\mathbf{1}_{\{\tau_{\mathrm{set}}[\hat{Q}(t)]_+>R_{\mathrm{liq}}\}}}{\hat{S}_{\mathrm{out}}(t)},
 \\
\dot{p}_{\mathrm{bill}} &= \rho p_{\mathrm{bill}} - \frac{\partial H}{\partial R_{\mathrm{bill}}} = (\rho - r_{\mathrm{bill}}) p_{\mathrm{bill}} + r_{\mathrm{bill}}, \\
\dot{p}_{\Delta P}
&= \rho p_{\Delta P} - \frac{\partial H}{\partial \Delta P}
= (\rho+\eta)\,p_{\Delta P} - 2c_{\mathrm{peg}}\,\Delta P,
\end{align}
where $\mathbf{1}_{\{\tau_{\mathrm{set}}[\hat{Q}(t)]_+>R_{\mathrm{liq}}\}}$ indicates a forecast cash shortfall over the settlement horizon.

The first-order conditions for control optimality yield
\begin{align}
0 &\in \partial_{\omega} H = \lambda_{\omega} \partial|\omega| + \rho_{\omega} \omega - p_{\mathrm{liq}} + p_{\mathrm{bill}}, \label{eq:omega-stationarity} \\
0 &= \frac{\partial H}{\partial \delta} = 2c_{\mathrm{fee}} \delta - \gamma p_{\Delta P}.
\end{align}
Thus, the fee control admits the linear feedback law
\begin{equation}
\delta^*(t) = \text{Proj}_{[-\delta_{\max}, \delta_{\max}]}\left(\frac{\gamma}{2c_{\mathrm{fee}}} p_{\Delta P}(t)\right),
\end{equation}
where projection operator $\text{Proj}_{[a,b]}(\cdot)$ clips values to the constraint interval. The reallocation control exhibits more sophisticated structure due to the $\ell_1$ trading cost (the $\lambda_{\omega}|\omega|$ term).

Define the switching function
\begin{equation}
S_{\omega}(t) = p_{\mathrm{bill}}(t) - p_{\mathrm{liq}}(t),
\end{equation}
which represents the instantaneous net opportunity cost of holding bills rather than cash.  
Solving the stationarity condition~\eqref{eq:omega-stationarity} for $\omega(t)$ yields
\begin{equation}
\omega^*(t) = \mathrm{Proj}_{[-\omega_{\max},\, \omega_{\max}]}
\left( -\frac{1}{\rho_{\omega}}\,\mathrm{shrink}\big(S_{\omega}(t), \lambda_{\omega}\big) \right),
\end{equation}
where $\mathrm{shrink}(z, \lambda) = \mathrm{sign}(z)\,\max\{|z| - \lambda,\, 0\}$ denotes the soft-thresholding operator.   This implies that the optimal portfolio reallocation control $\omega$ at time $t$ follows a \emph{threshold policy}:  
if the absolute marginal benefit of reallocation does not exceed the transaction cost $(|S_{\omega}(t)| \leq \lambda_{\omega})$, no reallocation within the reserve asset portfolio occurs $(\omega^*(t) = 0)$.  
If the marginal benefit exceeds the transaction-cost threshold $(|S_{\omega}(t)| > \lambda_{\omega})$, reallocation occurs with magnitude increasing in $|S_{\omega}(t)|-\lambda_{\omega}$ and capped at $\omega_{\max}$ by the projection, i.e.\ $0<|\omega^*(t)|\le \omega_{\max}$ with $\mathrm{sign}(\omega^*(t))=-\mathrm{sign}(S_{\omega}(t))$.

\subsection{Implementation with Settlement Windows}
\label{subsec:settlement}

\subsubsection{A Discretized Optimal Control Policy}
\label{subsubsec:discretized_ocp}

The continuous-time event-triggered policy, while theoretically elegant, requires potential control updates at every mint or redemption event—impractical given typical transaction volumes. We therefore develop a sampled-data formulation that aligns with operational realities and provides computational tractability.

Let $0 = \tau_0 < \tau_1 < \tau_2 < \cdots$ denote predetermined settlement or review times (e.g., Fedwire windows, hourly desk reviews) with intervals $\Delta_j = \tau_{j+1} - \tau_j$. We constrain controls to be piecewise constant over these windows:
\begin{equation}
\omega(t) \equiv \omega_j, \quad \delta(t) \equiv \delta_j \quad \text{for } t \in [\tau_j, \tau_{j+1}).
\end{equation}
This restriction reflects a real operational constraint, as treasury trades typically execute in batches at settlement windows, and defines the admissible control set $\mathcal{U}_{\mathrm{sd}}$.

Within this framework, the Hamiltonian contribution from reallocation over window $j$ becomes
\begin{equation}
\int_{\tau_j}^{\tau_{j+1}} \left(\lambda_{\omega}|\omega_j| + \tfrac{1}{2}\rho_{\omega}\omega_j^2 + (p_{\mathrm{bill}}(t) - p_{\mathrm{liq}}(t))\omega_j\right) dt.
\end{equation}
Define
\begin{equation}
S_j := \int_{\tau_j}^{\tau_{j+1}} (p_{\mathrm{bill}}(t) - p_{\mathrm{liq}}(t)) dt,
\label{eq:avg-switch}
\end{equation}
which represents the accumulated relative value of bills versus cash over the window.

\begin{theorem}[Window-Based Optimal Control]
The optimal sampled-data controls over window $j$ are:
\begin{align}
\omega_j^* &= \text{Proj}_{[-\omega_{\max}, \omega_{\max}]}\left(-\frac{1}{\rho_{\omega}\Delta_j} \text{shrink}(S_j, \lambda_{\omega}\Delta_j)\right), \label{eq:omega-window}\\
\delta_j^* &= \text{Proj}_{[-\delta_{\max}, \delta_{\max}]}\left(\frac{\gamma}{2c_{\mathrm{fee}}} \frac{1}{\Delta_j}\int_{\tau_j}^{\tau_{j+1}} p_{\Delta P}(t) dt\right),
\end{align}
where $\text{shrink}(z, \lambda) = \text{sign}(z)\max\{|z| - \lambda, 0\}$.
\end{theorem}

\begin{proof}
We need to minimize over window $j$ the functional
\begin{equation}
J_j(\omega_j) = \int_{\tau_j}^{\tau_{j+1}} \left(\lambda_{\omega}|\omega_j| + \tfrac{1}{2}\rho_{\omega}\omega_j^2 + (p_{\mathrm{bill}}(t) - p_{\mathrm{liq}}(t))\omega_j\right) dt
\end{equation}
subject to $\omega_j \in [-\omega_{\max}, \omega_{\max}]$. Since $\omega_j$ is constant over $[\tau_j, \tau_{j+1})$, we can factor it out:
\begin{equation}
J_j(\omega_j) = \lambda_{\omega}\Delta_j|\omega_j| + \tfrac{1}{2}\rho_{\omega}\Delta_j\omega_j^2 + S_j\omega_j,
\end{equation}
where $S_j = \int_{\tau_j}^{\tau_{j+1}} (p_{\mathrm{bill}}(t) - p_{\mathrm{liq}}(t)) dt$ as defined in \eqref{eq:avg-switch}.

First consider the unconstrained problem. The objective is convex (sum of convex functions), and strictly convex when $\rho_{\omega} > 0$ since the Hessian $\nabla^2 J_j = \rho_{\omega}\Delta_j > 0$.

For optimality, we require $0 \in \partial J_j(\omega_j)$, where $\partial$ denotes the subdifferential. Computing term by term:
\begin{itemize}
\item $\partial_{\omega_j}[\tfrac{1}{2}\rho_{\omega}\Delta_j\omega_j^2] = \rho_{\omega}\Delta_j\omega_j$
\item $\partial_{\omega_j}[-S_j\omega_j] = -S_j$
\item $\partial_{\omega_j}[\lambda_{\omega}\Delta_j|\omega_j|] = \lambda_{\omega}\Delta_j \cdot \partial|\omega_j|$
\end{itemize}

For the absolute value subdifferential at point $\omega_j$:
\begin{equation}
\partial|\omega_j| = \begin{cases}
\{1\} & \text{if } \omega_j > 0 \\
\{-1\} & \text{if } \omega_j < 0 \\
[-1, 1] & \text{if } \omega_j = 0
\end{cases}
\end{equation}

Therefore, the optimality condition $0 \in \partial J_j(\omega_j)$ becomes:
\begin{equation}
0 \in \lambda_{\omega}\Delta_j \cdot \partial|\omega_j| + \rho_{\omega}\Delta_j\omega_j + S_j.
\end{equation}

We analyze three cases:

\textit{Case 1: $\omega_j > 0$.} Then $\partial|\omega_j| = \{1\}$, so:
\begin{equation}
0 = \lambda_{\omega}\Delta_j + \rho_{\omega}\Delta_j\omega_j + S_j \implies \omega_j = \frac{-S_j - \lambda_{\omega}\Delta_j}{\rho_{\omega}\Delta_j}.
\end{equation}
This is valid only if $\omega_j > 0$, i.e., when $S_j < -\lambda_{\omega}\Delta_j$.

\textit{Case 2: $\omega_j < 0$.} Then $\partial|\omega_j| = \{-1\}$, so:
\begin{equation}
0 = -\lambda_{\omega}\Delta_j + \rho_{\omega}\Delta_j\omega_j + S_j \implies \omega_j = \frac{-S_j + \lambda_{\omega}\Delta_j}{\rho_{\omega}\Delta_j}.
\end{equation}
This is valid only if $\omega_j < 0$, i.e., when $S_j > \lambda_{\omega}\Delta_j$.

\textit{Case 3: $\omega_j = 0$.} Then $\partial|\omega_j| = [-1, 1]$, so we need:
\begin{equation}
0 \in \lambda_{\omega}\Delta_j[-1, 1] + 0 + S_j = [-\lambda_{\omega}\Delta_j + S_j, \lambda_{\omega}\Delta_j + S_j].
\end{equation}
This holds if and only if $-\lambda_{\omega}\Delta_j \leq S_j \leq \lambda_{\omega}\Delta_j$, i.e., $|S_j| \leq \lambda_{\omega}\Delta_j$.

Combining all cases, the unconstrained optimizer is:
\begin{equation}
\omega_j^{\text{unc}} = \begin{cases}
\frac{-S_j - \lambda_{\omega}\Delta_j}{\rho_{\omega}\Delta_j} & \text{if } S_j < -\lambda_{\omega}\Delta_j \\[4pt]
0 & \text{if } |S_j| \leq \lambda_{\omega}\Delta_j \\[4pt]
\frac{-S_j + \lambda_{\omega}\Delta_j}{\rho_{\omega}\Delta_j} & \text{if } S_j > \lambda_{\omega}\Delta_j
\end{cases}
\end{equation}

This can be written compactly using the soft-thresholding operator:
\begin{equation}
\omega_j^{\text{unc}} = -\frac{1}{\rho_{\omega}\Delta_j}\text{shrink}(S_j, \lambda_{\omega}\Delta_j),
\end{equation}
where $\text{shrink}(z, \lambda) := \text{sign}(z)\max\{|z| - \lambda, 0\}$. 

The constraint set $[-\omega_{\max}, \omega_{\max}]$ is closed and convex. Since $J_j$ is strictly convex, the constrained minimizer is unique and given by the Euclidean projection:
\begin{equation}
\omega_j^* = \text{Proj}_{[-\omega_{\max}, \omega_{\max}]}(\omega_j^{\text{unc}}) = \text{Proj}_{[-\omega_{\max}, \omega_{\max}]}\left(-\frac{1}{\rho_{\omega}\Delta_j}\text{shrink}(S_j, \lambda_{\omega}\Delta_j)\right),
\end{equation}
which completes the proof.
\end{proof}

This result preserves the essential structure of the continuous-time solution while offering practical advantages. Control decisions require only time-averaged costates $S_j$ and $\bar{p}_{\Delta P,j}$, computed once per settlement window rather than at each transaction, reducing computational complexity from potentially thousands of transactions to typically 6-24 daily windows. The policy naturally aligns with institutional constraints by executing treasury transactions at predetermined windows matching market infrastructure like Fedwire hours and repo settlement times, while the averaged switching function $S_j$ smooths intra-window flow volatility to prevent overreaction to temporary imbalances. The decision rule admits a transparent interpretation: rebalance at each window if the averaged value differential exceeds the transaction cost threshold $\lambda_{\omega}\Delta_j$, otherwise maintain positions. Moreover, the effective inaction threshold automatically adapts to operating conditions, scaling with window length to allow more aggressive rebalancing over longer intervals where costs amortize better, while becoming more conservative during stress periods with frequent reviews.

\subsubsection{Monotone Stress-Response Behavior}
\label{subsubsec:monotone_stress_response}

In this section, we analyze how the window-based optimal control responds to varying levels of redemption pressure. The system favors yield-bearing Treasury holdings during calm periods and shifts toward cash as redemption demand intensifies. The following theorem establishes this behavior through analysis of the costate dynamics:
\begin{theorem}[Monotone Stress-Threshold Principle]
\label{thm:stress-threshold}
Consider the window-based optimal control \eqref{eq:omega-window} over interval $[\tau_j, \tau_{j+1})$. Let $\hat{Q}^{(1)}$ and $\hat{Q}^{(2)}$ be two forecast paths with $\hat{Q}^{(1)}(t) \geq \hat{Q}^{(2)}(t)$ for all $t \in [\tau_j, \tau_{j+1})$. Then:
\begin{enumerate}
\item The optimal reallocation satisfies $\omega_j^*(\hat{Q}^{(1)}) \leq \omega_j^*(\hat{Q}^{(2)})$
\item For any parametric family $\hat{Q}_{\alpha}(t) = \alpha q(t)$ with $q(t) \geq 0$, there exists a threshold $\alpha^* \geq 0$ such that:
\begin{equation}
\omega_j^*(\hat{Q}_{\alpha}) = \begin{cases}
0 & \text{if } \alpha \leq \alpha^* \\
\neq 0 & \text{if } \alpha > \alpha^*
\end{cases}
\end{equation}
\item As $\alpha \to \infty$, $\omega_j^*(\hat{Q}_{\alpha}) \to -\omega_{\max}$ (maximum cash building)
\end{enumerate}
\end{theorem}

\begin{proof}
We establish monotonicity through a sequence of comparison results on the state-costate system.

\textbf{Step 1: State monotonicity.} Under fixed controls $(\omega_j, \delta_j)$ and identical initial conditions, the state dynamics yield:
\begin{equation}
\frac{d}{dt}[R_{\mathrm{liq}}^{(1)} - R_{\mathrm{liq}}^{(2)}] = -[\hat{Q}^{(1)} - \hat{Q}^{(2)}] \leq 0,
\end{equation}
implying $R_{\mathrm{liq}}^{(1)}(t) \leq R_{\mathrm{liq}}^{(2)}(t)$ for all $t \in [\tau_j, \tau_{j+1})$.

\noindent 
For the peg deviation, the fair-value deviations are
\[
\Delta P_{\mathrm{fair}}^{(i)}(t)
:=
\mathbf{1}_{\{\hat{S}_{\mathrm{out}}(t)>0\}}
\left[
\frac{\tau_{\mathrm{set}}[\hat{Q}^{(i)}(t)]_+ - R_{\mathrm{liq}}^{(i)}(t)}
     {\hat{S}_{\mathrm{out}}(t)}
\right]_0^1,
\qquad i\in\{1,2\}.
\]
Since $\hat{Q}^{(1)}(t)\ge\hat{Q}^{(2)}(t)$ and
$R_{\mathrm{liq}}^{(1)}(t)\le R_{\mathrm{liq}}^{(2)}(t)$, the numerator
$\tau_{\mathrm{set}}[\hat{Q}^{(i)}]_+ - R_{\mathrm{liq}}^{(i)}$ is
weakly larger for $i=1$ than for $i=2$, and the $[\cdot]_0^1$ clipping
preserves this ordering, giving
$\Delta P_{\mathrm{fair}}^{(1)}(t)\ge \Delta P_{\mathrm{fair}}^{(2)}(t)$.
Under identical fee controls $\delta_j$, the peg dynamics imply that the difference
$D(t):=\Delta P^{(1)}(t)-\Delta P^{(2)}(t)$ satisfies
\begin{equation}
\dot D(t)
= \eta\big(\Delta P_{\mathrm{fair}}^{(1)}(t)-\Delta P_{\mathrm{fair}}^{(2)}(t)\big) - \eta D(t).
\end{equation}
With $D(\tau_j)=0$ and nonnegative forcing, comparison yields
$D(t)\ge 0$ for all $t\in[\tau_j,\tau_{j+1})$, i.e.\ $\Delta P^{(1)}(t)\ge \Delta P^{(2)}(t)$.

\textbf{Step 2: Costate monotonicity.} Working backwards from terminal conditions at $\tau_{j+1}$, the variation of constants formula for the peg costate gives:
\begin{equation}
p_{\Delta P}(t)
=
e^{(\rho+\eta)(t-\tau_{j+1})}p_{\Delta P}(\tau_{j+1})
+ 2c_{\mathrm{peg}}\int_t^{\tau_{j+1}} e^{(\rho+\eta)(t-s)}\Delta P(s)\,ds.
\end{equation}
Since $\Delta P^{(1)} \geq \Delta P^{(2)}$, we have $p_{\Delta P}^{(1)}(t) \geq p_{\Delta P}^{(2)}(t)$.

For the liquidity costate, the dynamics
\begin{equation}
\dot{p}_{\mathrm{liq}}
=
(\rho - r_{\mathrm{cash}})p_{\mathrm{liq}}
+ \eta\,
\frac{\mathbf{1}_{\{\tau_{\mathrm{set}}[\hat{Q}(t)]_+ > R_{\mathrm{liq}}\}}}
     {\hat{S}_{\mathrm{out}}}\,p_{\Delta P}
\end{equation}
form a linear system with monotone coefficient (the indicator function) and monotone forcing (since $p_{\Delta P}^{(1)} \geq p_{\Delta P}^{(2)}$ and $p_{\Delta P}\ge 0$ for $\Delta P\ge 0$). Hence,
\[
\dot{p}_{\mathrm{liq}}^{(1)}(t) \leq \dot{p}_{\mathrm{liq}}^{(2)}(t) \quad \text{for all } t,
\]
and, consequently,
\[
p_{\mathrm{liq}}^{(1)}(t) \leq p_{\mathrm{liq}}^{(2)}(t) \quad \text{for all } t.
\]

The bill costate evolution is independent of $\hat{Q}$, so $p_{\mathrm{bill}}^{(1)} = p_{\mathrm{bill}}^{(2)}$.

\textbf{Step 3: Switching function monotonicity.} The averaged switching function satisfies:
\begin{equation}
S_j(\hat{Q}^{(1)}) - S_j(\hat{Q}^{(2)}) = \int_{\tau_j}^{\tau_{j+1}} [p_{\mathrm{liq}}^{(2)}(t) - p_{\mathrm{liq}}^{(1)}(t)]\,dt \geq 0.
\end{equation}

From \eqref{eq:omega-window}, the optimal control is:
\begin{equation}
\omega_j^* = \text{Proj}_{[-\omega_{\max}, \omega_{\max}]}\left(-\frac{1}{\rho_{\omega}\Delta_j}\text{shrink}(S_j, \lambda_{\omega}\Delta_j)\right).
\end{equation}
Since $z \mapsto -\text{shrink}(z, \lambda)$ is monotone decreasing and projection preserves monotonicity, larger $S_j$ yields smaller (more negative) $\omega_j^*$.

For the parametric family $\hat{Q}_{\alpha} = \alpha q$, since $S_j(\hat{Q}_{\alpha})$ is continuous in the forecast path $\hat{Q}_{\alpha}(t)$, it is also continuous in $\alpha$ for $\hat{Q}_{\alpha}(t) = \alpha q(t)$. The threshold structure of the shrinkage operator then guarantees existence of $\alpha^*$ where the control switches from inaction to action. 
\end{proof}

This theorem provides rigorous justification for a simple operational rule:
\begin{itemize}
\item \textbf{Normal markets} (low $\alpha$): The optimal policy maintains treasury positions to capture yield
\item \textbf{Stress onset} ($\alpha \approx \alpha^*$): The system reaches a tipping point where liquidity concerns dominate yield considerations
\item \textbf{Crisis mode} (high $\alpha$): The policy aggressively builds cash reserves at maximum feasible rate
\end{itemize}

The monotonicity property ensures this transition is predictable and one-directional as stress intensifies, preventing rapid switching that could amplify market disruption and erode confidence. This result provides a mathematical formalization of prudent reserve management: Stable markets allow for large positions in interest-bearing T-bonds. With increasing volatility, a significant shift to cash is necessary before sell-offs become self-reinforcing.

\section{Numerical Results}
\label{sec:numerical_results}

In this section, we use simulations of diverse market phenomena to demonstrate the effectiveness of the proposed control algorithm considering both profitability and safety.

\subsection{Experimental Design}
\label{subsec:exp_design}

We conduct numerical experiments over 92 calendar days (August 1 - October 31, 2025) with three equally spaced settlement windows per day (so $\Delta_j=8$ hours), aligning with major settlement times. The performance of various control strategies is evaluated using three scenarios that capture characteristic market stress patterns empirically observed in financial markets. The construction of these scenarios enables precise testing of our central hypothesis that the developed control algorithm fulfills three essential requirements: first, the exploitation of lucrative investment opportunities in Treasury Bills during calm market phases; second, the maintenance of sufficient liquidity reserves in volatile phases with potentially endogenous bank-run dynamics; and third, the early detection of emerging volatility phases, such that proactive liquidity procurement is initiated at early signs of a demand shock. To systematically investigate these properties, we define the following market scenarios:

\begin{enumerate}
\item \textbf{Single shock event:} The system operates under baseline conditions for days 0-30, experiences elevated redemption pressure during days 30-45 with baseline intensity $\lambda_{0,R}$ multiplied by a factor $\chi \sim \text{Uniform}[2,4]$, then undergoes exponential mean-reversion with a 4-day time constant through day 92. This stress test examines the controller's ability to respond decisively to sustained redemption pressure and subsequently restore normalized operations.

\item \textbf{Prolonged clustering:} The system maintains baseline parameters until day 30. From days 30-70, we increase the self-excitation parameter $\kappa_R$ equally for both mint and redemption requests to bring the branching ratio $\kappa_R/\theta_R$ to 0.85, approaching criticality. This near-critical regime generates persistent "aftershock" bursts in mints and redemptions that emerge despite normal baseline intensities. After day 70, parameters return to baseline.

\item \textbf{False alarm (whipsaw):} A brief mint surge occurs during days 30-32 where $\lambda_{0,M}$ triples, followed by a modest redemption increase to $1.5 \times \lambda_{0,R}$ during days 32-35, then returning to baseline. This pattern tests whether the controller avoids unnecessary cash accumulation in response to transient imbalances.
\end{enumerate}

Peg stability is measured by the failure probability, defined as the relative frequency of severe depeg events ($\Delta P \geq 0.85$). Responsiveness is measured by the time lag between a shock onset and the controller's increase in cash-buying speed.

In addition to the optimal control described in Subsection \ref{subsec:settlement}, we evaluate the performance of three reference strategies:
\begin{enumerate}
\item \textbf{The maximum yield strategy:} Maximizes Treasury Bond holdings and converts these to cash exclusively when immediately necessary to cover redemption requests.
\item \textbf{The maximum liquidity strategy:} Maintains the entire portfolio exclusively in liquid assets.
\item {\textbf{The equal weights strategy:} Maintains a fixed 50/50 allocation between cash and Treasury Bills, rebalancing at each settlement window to restore the target split and converting bills to cash when immediately necessary to cover redemption requests.}
\end{enumerate}

\subsection{Parameter Estimation}
\label{subsec:parameter_estimation}

Throughout the numerical section we measure time in hours. Annualized rates are converted to hourly continuously-compounded rates by division by $8760$:
\begin{itemize}
\item 90-day Treasury-bill yield: 
$
r_{\mathrm{bill}} = 0.054/8760 \approx 6.16 \times 10^{-6}\ \mathrm{hour}^{-1}.
$
(Over an $8$-hour settlement window this corresponds to $r_{\mathrm{bill}}\Delta_j \approx 4.93\times 10^{-5}$.)
\item Discount rate: 
$
\rho = 0.08/8760 \approx 9.13 \times 10^{-6}\ \mathrm{hour}^{-1}.
$
(Over an $8$-hour window: $\rho \Delta_j \approx 7.31\times 10^{-5}$.)
\end{itemize}

The following operational parameters are calibrated to approximate those commonly found in industry practices and settlement infrastructure:
\begin{itemize}
\item Maximum reallocation rate: $\omega_{\max}=0.1\,S_{\mathrm{out}}(0)/\Delta_j$, permitting up to 10\% of the initial outstanding supply to be shifted per settlement window.
\item Settlement windows: $\Delta_j = 8$ hours (three windows per day)
\item Settlement horizon in \eqref{eq:DeltaPfair-def}: $\tau_{\mathrm{set}} = \Delta_j = 8$ hours.
\item Fee bounds (even if $\delta\equiv 0$ in the baseline experiments): $\delta_{\max}=0.01$ (100 bps).
\end{itemize}

To simplify the interpretation of numerical results, we assume zero mint or redemption fee ($\delta = 0$). For modelling the Hawkes processes, we choose parameters to mimic the dynamics observed in major stablecoins (USDC, USDT) during 2024-2025. For baseline intensities, we set $\lambda_{0,R} = 80$ events/hour for redemptions and $\lambda_{0,M} = 85$ events/hour for mints, corresponding to daily baseline volumes of 1,920 and 2,040 events respectively, which align with typical median volumes for these assets. 

To capture the clustering behavior characteristic of stablecoin transactions, we parameterize the self-excitation dynamics with decay rates of $\theta_R = 2.0$ hour$^{-1}$ and $\theta_M = 1.5$ hour$^{-1}$, yielding half-lives of approximately 20-30 minutes that reflect how quickly bursts of activity dissipate. We set the corresponding jump sizes to $\kappa_R = 0.8$ and $\kappa_M = 0.6$, which ensure subcritical branching ratios ($\kappa_R/\theta_R = 0.4 < 1$ and $\kappa_M/\theta_M = 0.4 < 1$) necessary for process stability. Finally, transaction sizes follow exponential distributions with mean redemption size $\bar{z}_R = \$250,000$ and mean mint size $\bar{z}_M = \$300,000$, values representative of institutional trading volumes in these markets.

The cost function parameters are scaled to be commensurate with hourly dollar cash flows. 
For peg stability we penalize squared deviations and choose $c_{\mathrm{peg}}$ such that sustaining a deviation of $\Delta P_{\mathrm{tol}}=10$ bps ($=0.001$) carries an hourly cost comparable to the hourly T-bill carry on the outstanding supply:
\begin{equation}
c_{\mathrm{peg}}
=
\frac{r_{\mathrm{bill}}\,S_{\mathrm{out}}(0)}{\Delta P_{\mathrm{tol}}^{2}}
\approx
\frac{(6.16\times 10^{-6})\cdot (10^{10})}{(10^{-3})^{2}}
\approx 6.16\times 10^{10}.
\end{equation}
For fee frictions we choose $c_{\mathrm{fee}}$ such that charging $\delta_{\mathrm{tol}}=10$ bps on a typical hour of redemptions incurs a quadratic penalty comparable to the corresponding fee revenue:
\begin{equation}
c_{\mathrm{fee}}
=
\frac{\lambda_{0,R}\bar z_R}{\delta_{\mathrm{tol}}}
\approx
\frac{(100)\cdot(2.5\times 10^{5})}{10^{-3}}
\approx 2.5\times 10^{10}.
\end{equation}

Trading costs reflect treasury-market execution frictions. We set the linear term to a 1bp half-spread,
$\lambda_{\omega}=10^{-4}$, so that $\int \lambda_{\omega}|\omega(t)|\,dt$ equals approximately 1bp times traded notional,
and $\rho_{\omega}=8\times 10^{-13}$ (in hour/\$ units) for the quadratic term. The peg adjustment speed is set to $\eta=\ln(2)/4\approx 0.173\ \mathrm{hour}^{-1}$, corresponding to a half-life of 4 hours for the gap between $\Delta P(t)$ and $\Delta P_{\mathrm{fair}}(t)$. 
For the fee effectiveness we set $\gamma = 5.0\ \mathrm{hour}^{-1}$.
Finally, we set the peg-to-redemption feedback strength $\zeta$ to 2000.

For initial conditions, we consider a system at scale with outstanding supply $S_{\mathrm{out}}(0) = \$10$ billion and zero deviation from the peg ($\Delta P(0) = 0$). Transaction intensities begin at baseline levels: $\lambda_R(0) = \lambda_{0,R}$ and $\lambda_M(0) = \lambda_{0,M}$. The reserve allocation strategies differ between controllers: both the optimal controller and the maximum yield controller begin with $R_{\mathrm{liq}}(0) = \$1$ billion (10\%) in liquid assets and $R_{\mathrm{bill}}(0) = \$9$ billion (90\%) in Treasury bills, while the maximum liquidity controller starts with all reserves (\$10 billion) allocated to liquid cash.

\subsection{Presentation Of Results}
\label{subsec:results}

Revenue is calculated as the negative of the cost functional. Table \ref{tab:cost} shows the total revenue by scenario and control strategy. Table \ref{tab:freq-resp} displays the fallout rates and responsiveness measures. We simulated 100 scenarios and report average values for all metrics, except for the fallout rate, which is shown as a relative frequency. Figures \ref{fig:peg-combined}--\ref{fig:peg-combined-fa} display the peg deviations, and Figures \ref{fig:controls-grid}--\ref{fig:controls-grid-fa} illustrate the trajectories of the control input $\omega$ for all analyzed control strategies and scenarios.

\begin{table}[H]
\centering
\begin{tabular}{l l c}
\toprule
\multicolumn{1}{c}{Strategy} & \multicolumn{1}{c}{Scenario}
  & \multicolumn{1}{c}{Total Revenue [\$M]} \\
\midrule
\multirow{3}{*}{Optimal Control}
  & Single Shock            & 98.02 \\
  & Prolonged Clustering    & 421.34 \\
  & False Alarm             & 300.79 \\
\addlinespace
\multirow{3}{*}{Maximum Yield}
  & Single Shock            & 83.17 \\
  & Prolonged Clustering    & 100.61 \\
  & False Alarm             & 296.41 \\
\addlinespace
\multirow{3}{*}{Maximum Liquidity}
  & Single Shock            & 0 \\
  & Prolonged Clustering    & 0 \\
  & False Alarm             & 0 \\
\addlinespace
\multirow{3}{*}{Equal Weights}
  & Single Shock            & 51.03 \\
  & Prolonged Clustering    & 213.33 \\
  & False Alarm             & 151.03 \\
\bottomrule
\end{tabular}
\caption{Avg.\ Revenue by Strategy and Scenario}
\label{tab:cost}
\end{table}

\begin{table}[H]
\centering
\begin{tabular}{l l c c}
\toprule
\multicolumn{1}{c}{Strategy} & \multicolumn{1}{c}{Scenario}
  & \multicolumn{1}{c}{Depeg Frequency} & \multicolumn{1}{c}{Resp. (days)} \\
\midrule
\multirow{3}{*}{Optimal Control}
  & Single Shock            & 0 & 2.77 \\
  & Prolonged Clustering    & 0 & 0.67 \\
  & False Alarm             & 0 & n/a \\
\addlinespace
\multirow{3}{*}{Maximum Yield}
  & Single Shock            & 1 & 0.33 \\
  & Prolonged Clustering    & 1 & 0.33 \\
  & False Alarm             & 0 & n/a \\
\addlinespace
\multirow{3}{*}{Maximum Liquidity}
  & Single Shock            & 0 & n/a \\
  & Prolonged Clustering    & 0 & n/a \\
  & False Alarm             & 0 & n/a \\
\addlinespace
\multirow{3}{*}{Equal Weights}
  & Single Shock            & 0.69 & 0.33 \\
  & Prolonged Clustering    & 0 &  6.00 \\ 
  & False Alarm             & 0 &   n/a\\
\bottomrule
\end{tabular}
\caption{Depeg Frequency and Responsiveness}
\label{tab:freq-resp}
\end{table}

\begin{figure}[htbp]
  \centering
  \includegraphics[width=\textwidth]{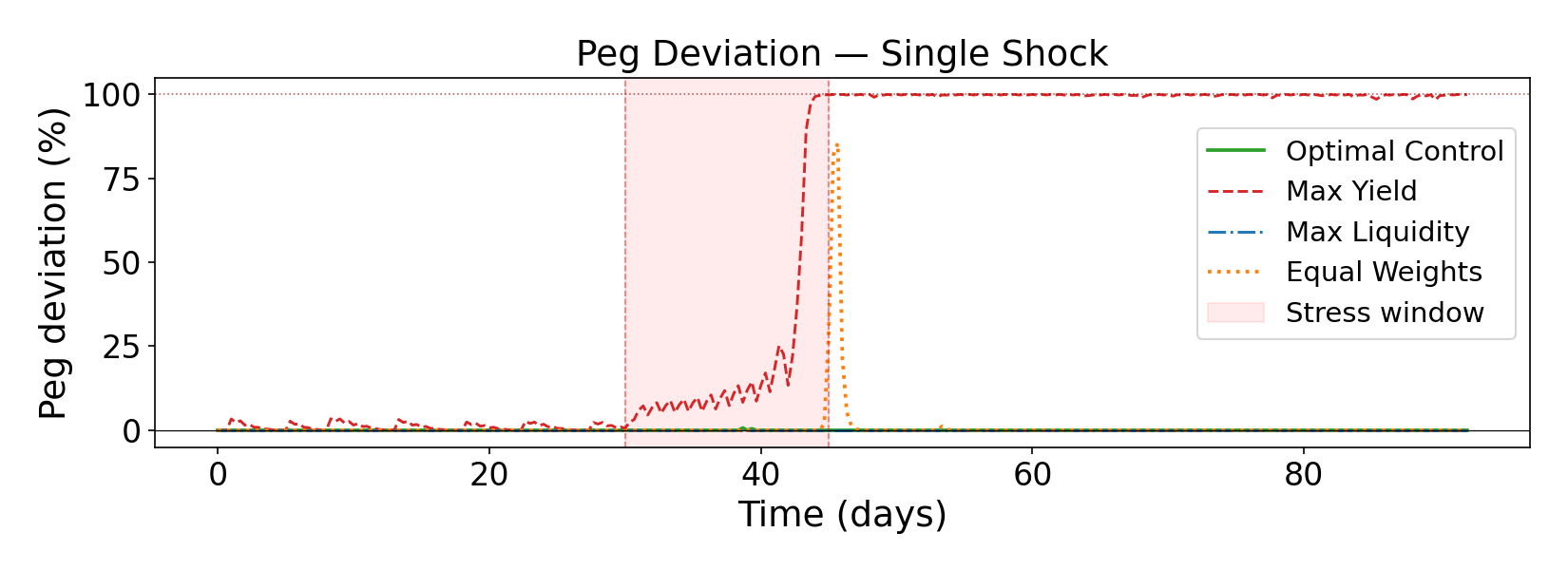}
  \caption{Peg deviation $\Delta P(t)$ under the single-shock scenario for
           all four control strategies.  The shaded region marks the stress
           window (days 30-45).}
  \label{fig:peg-combined}
\end{figure}

\begin{figure}[htbp]
  \centering
  \includegraphics[width=\textwidth]{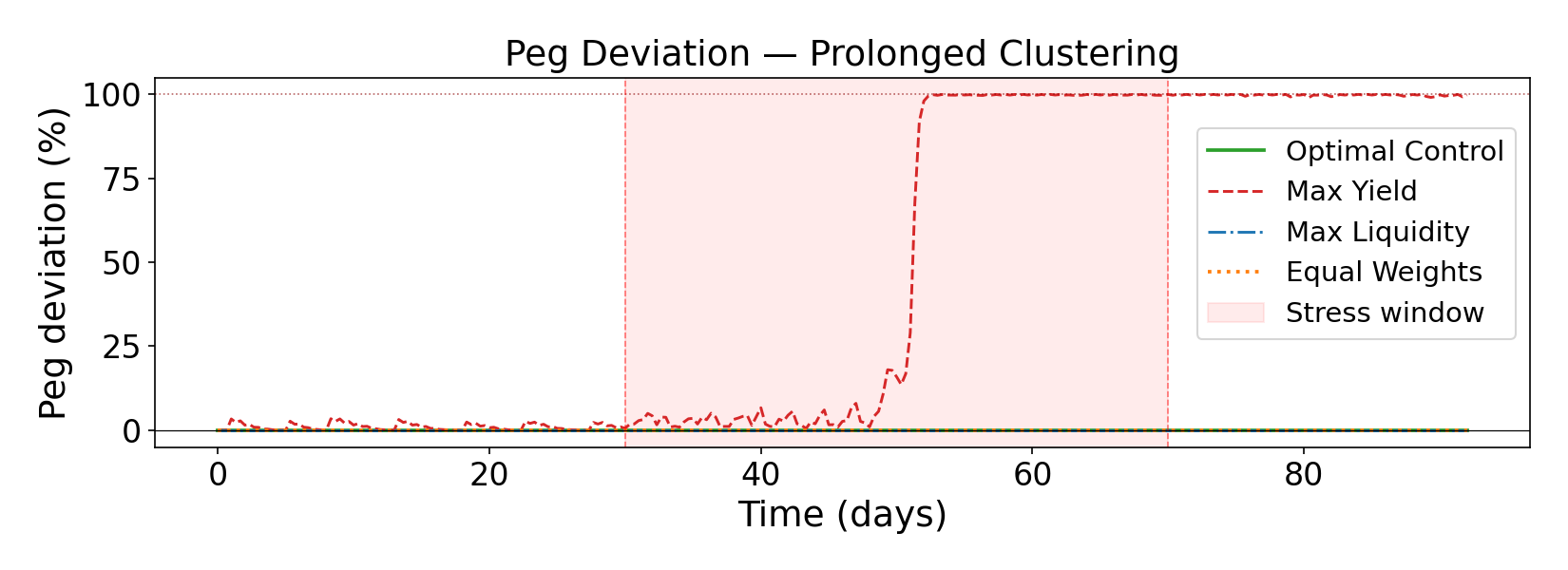}
  \caption{Peg deviation $\Delta P(t)$ under the prolonged-clustering
           scenario for all four control strategies.  The shaded region
           marks the stress window (days 30--70).}
  \label{fig:peg-combined-clustering}
\end{figure}

\begin{figure}[htbp]
  \centering
  \includegraphics[width=\textwidth]{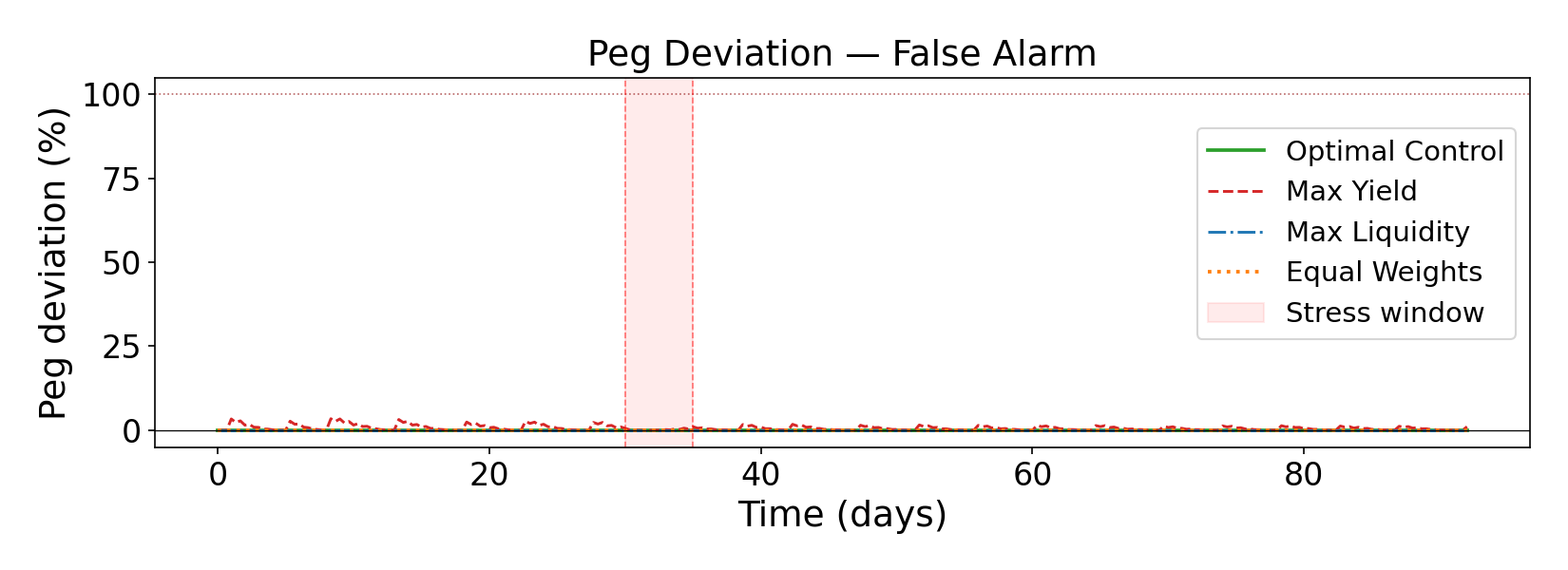}
  \caption{Peg deviation $\Delta P(t)$ under the false-alarm scenario for
           all four control strategies.  The shaded region marks the stress
           window (days 30--35).}
  \label{fig:peg-combined-fa}
\end{figure}

\begin{figure}[htbp]
  \centering
  \begin{tabular}{@{}cc@{}}
    \includegraphics[width=0.49\textwidth]{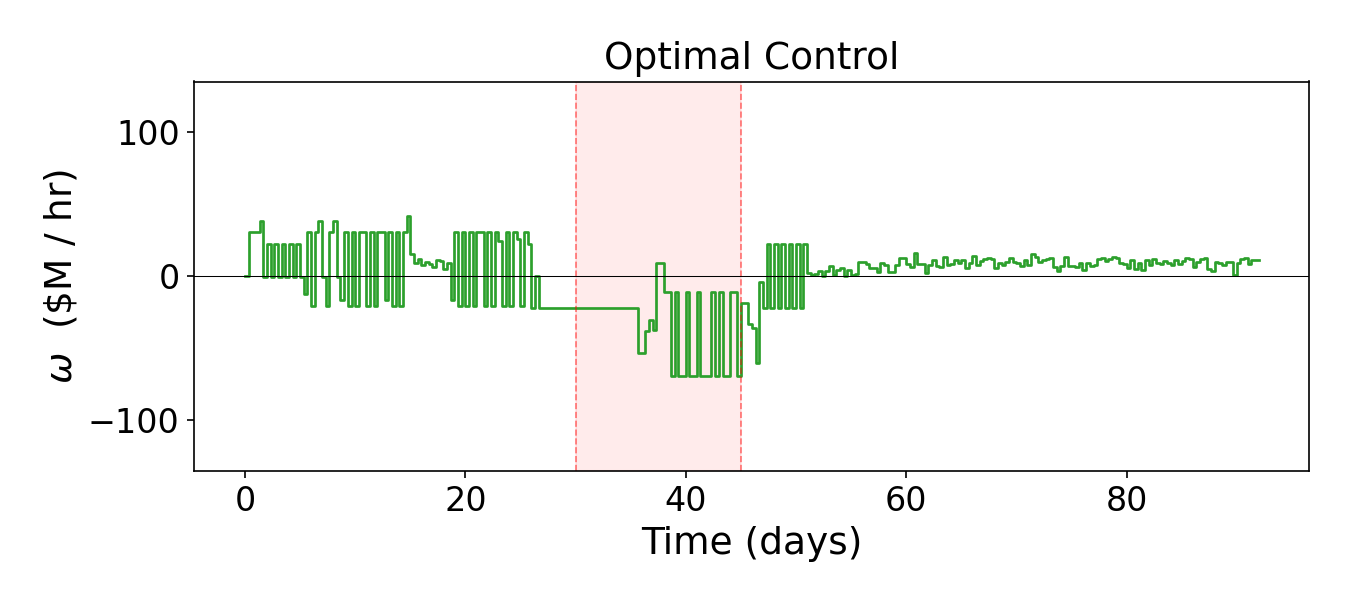} &
    \includegraphics[width=0.49\textwidth]{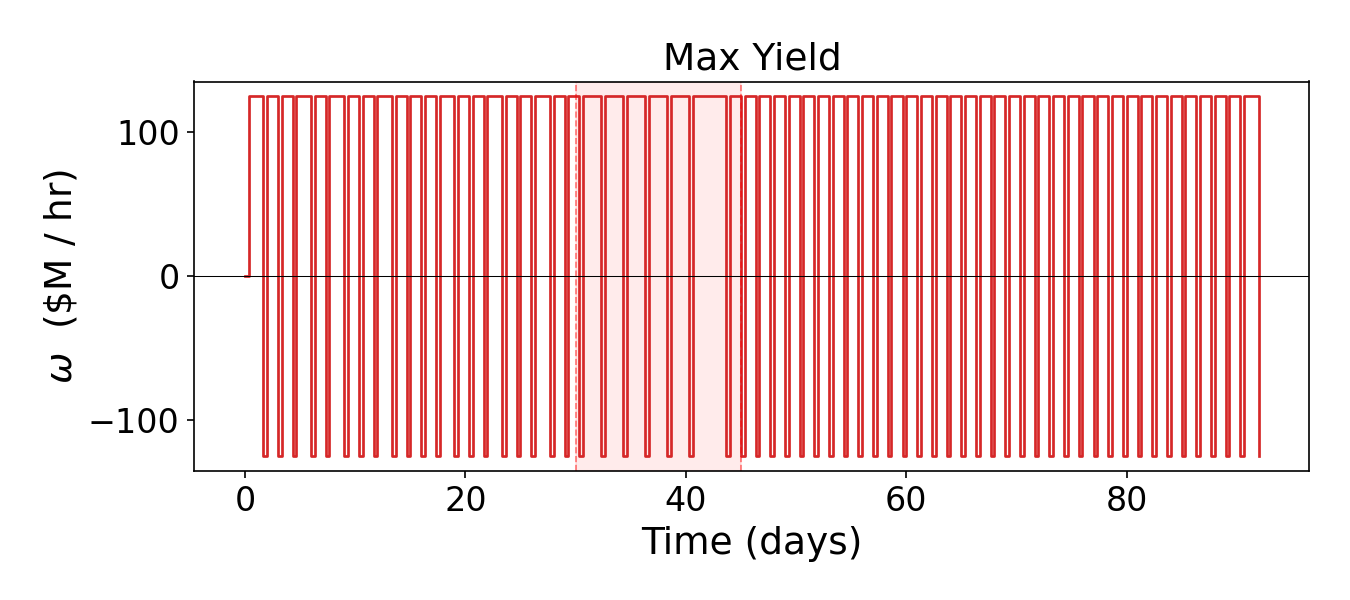} \\[4pt]
    \includegraphics[width=0.49\textwidth]{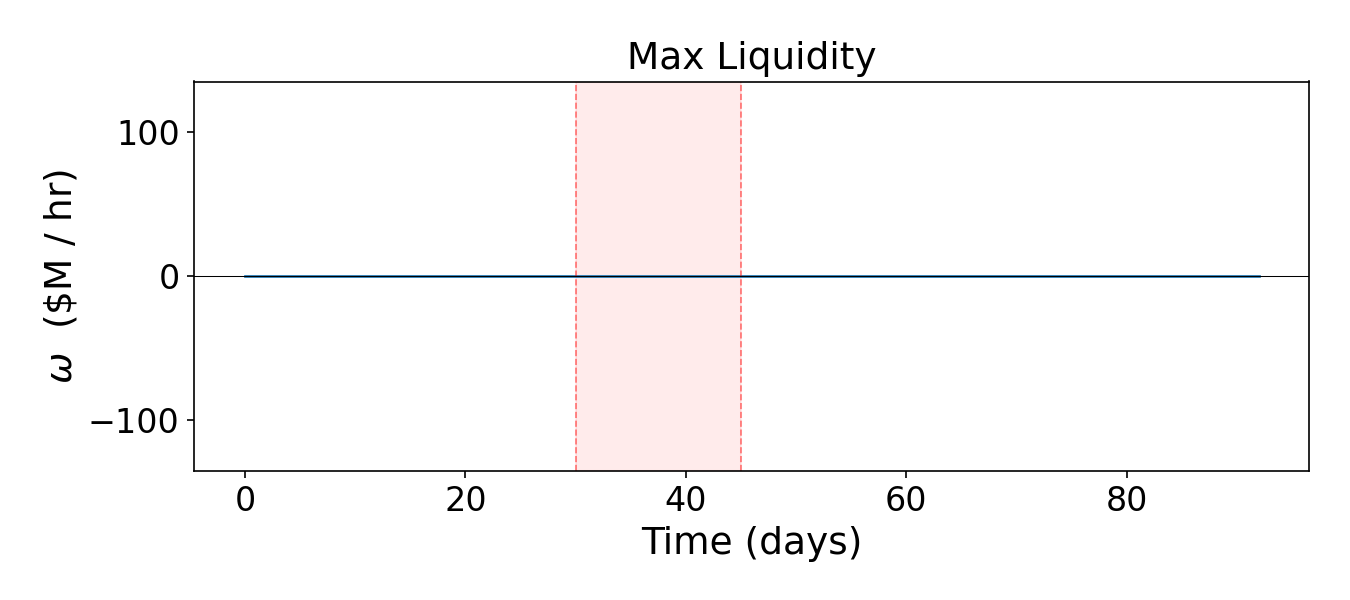} &
    \includegraphics[width=0.49\textwidth]{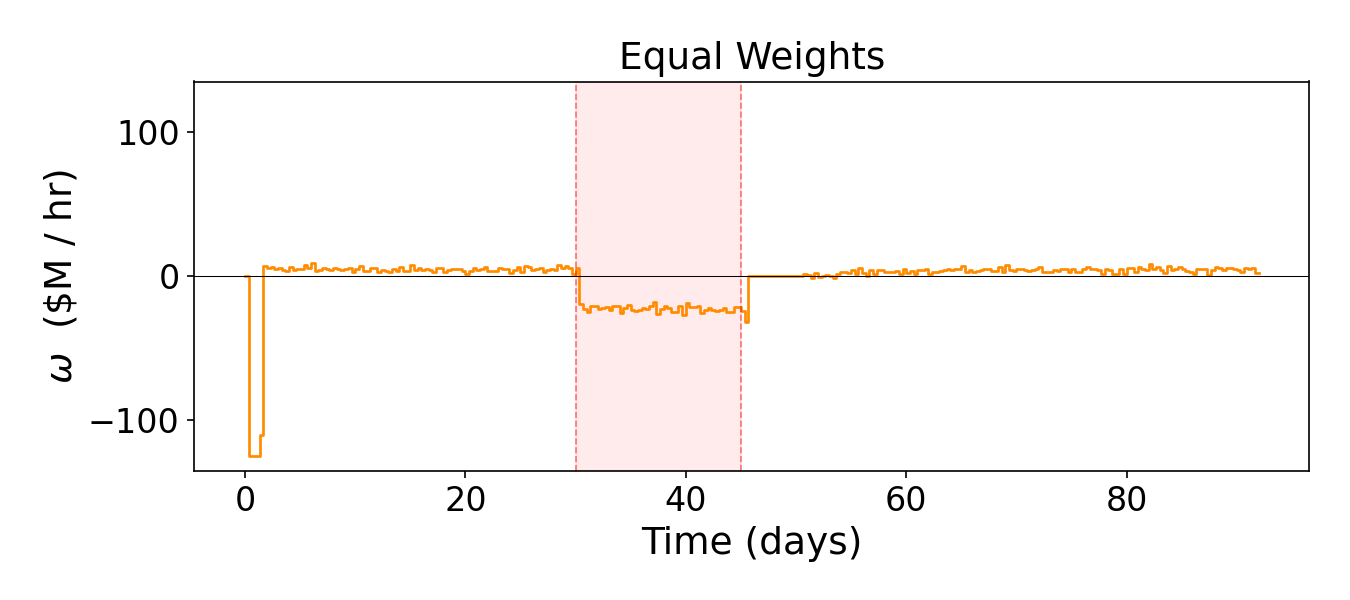}
  \end{tabular}
  \caption{Reallocation rate $\omega(t)$ under the single-shock scenario for all four strategies. }
\label{fig:controls-grid}
\end{figure}

\begin{figure}[htbp]
  \centering
  \begin{tabular}{@{}cc@{}}
    \includegraphics[width=0.49\textwidth]{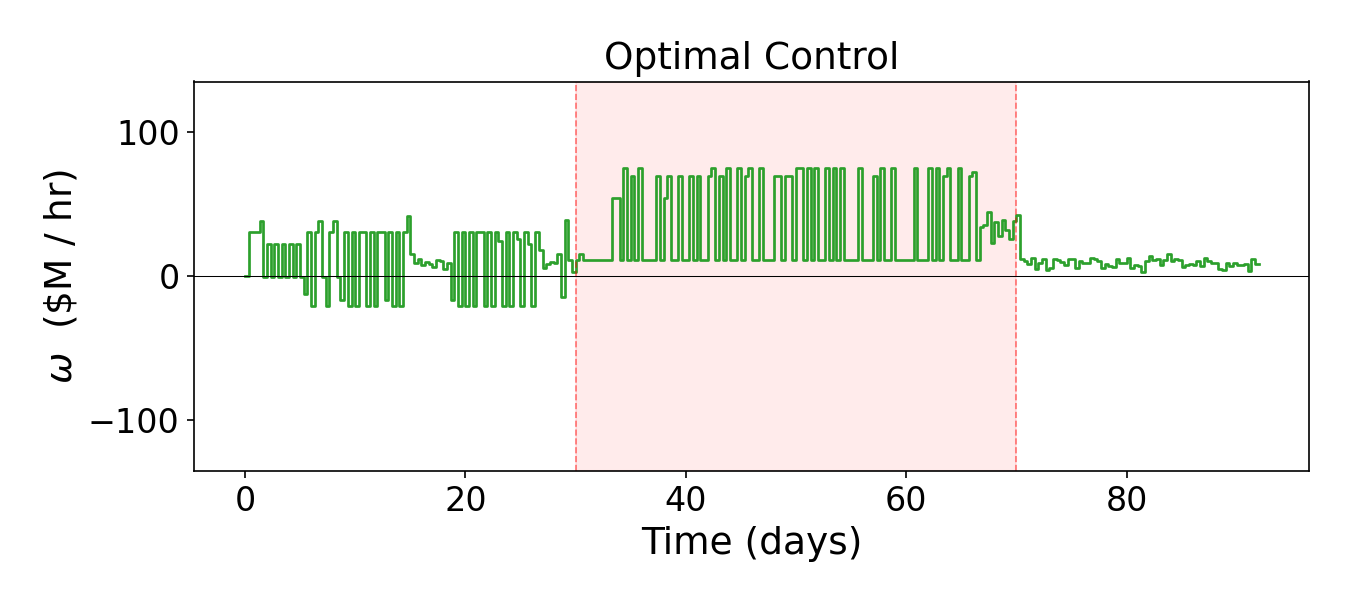} &
    \includegraphics[width=0.49\textwidth]{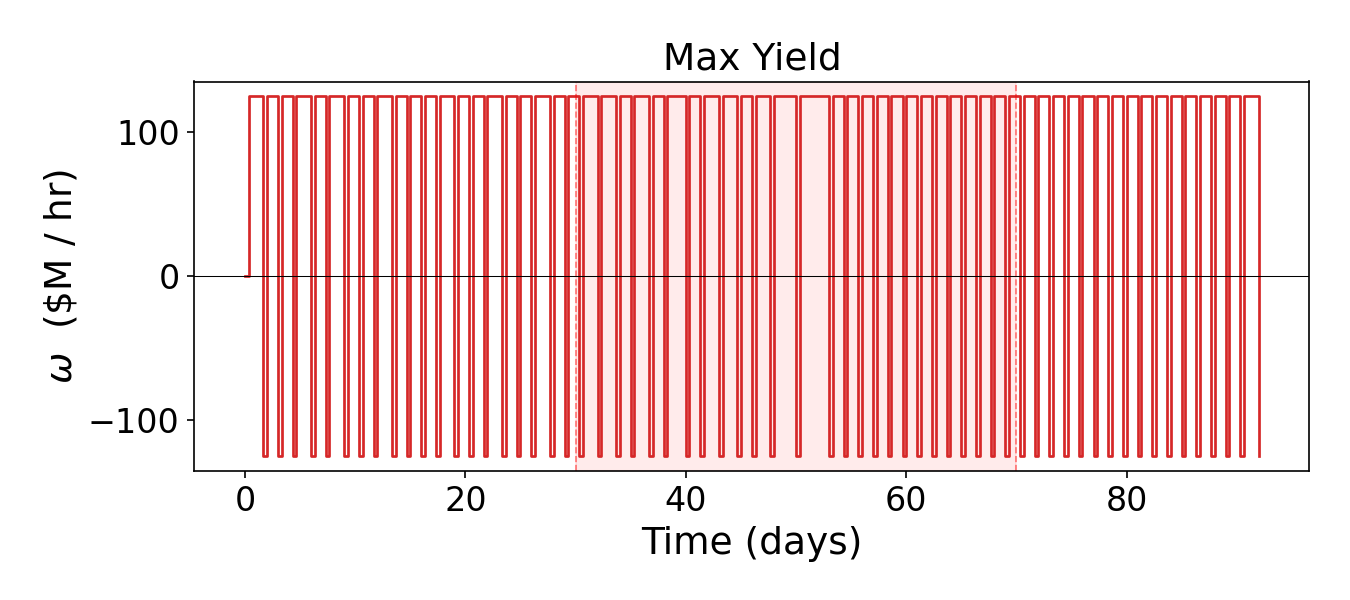} \\[4pt]
    \includegraphics[width=0.49\textwidth]{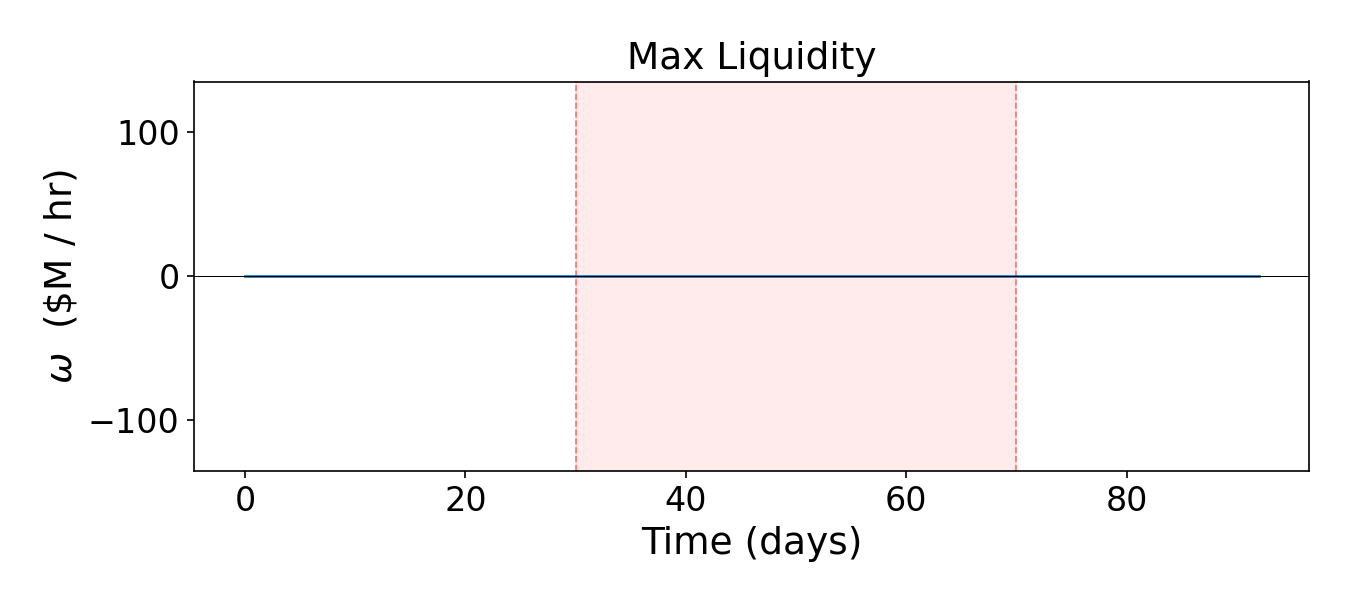} &
    \includegraphics[width=0.49\textwidth]{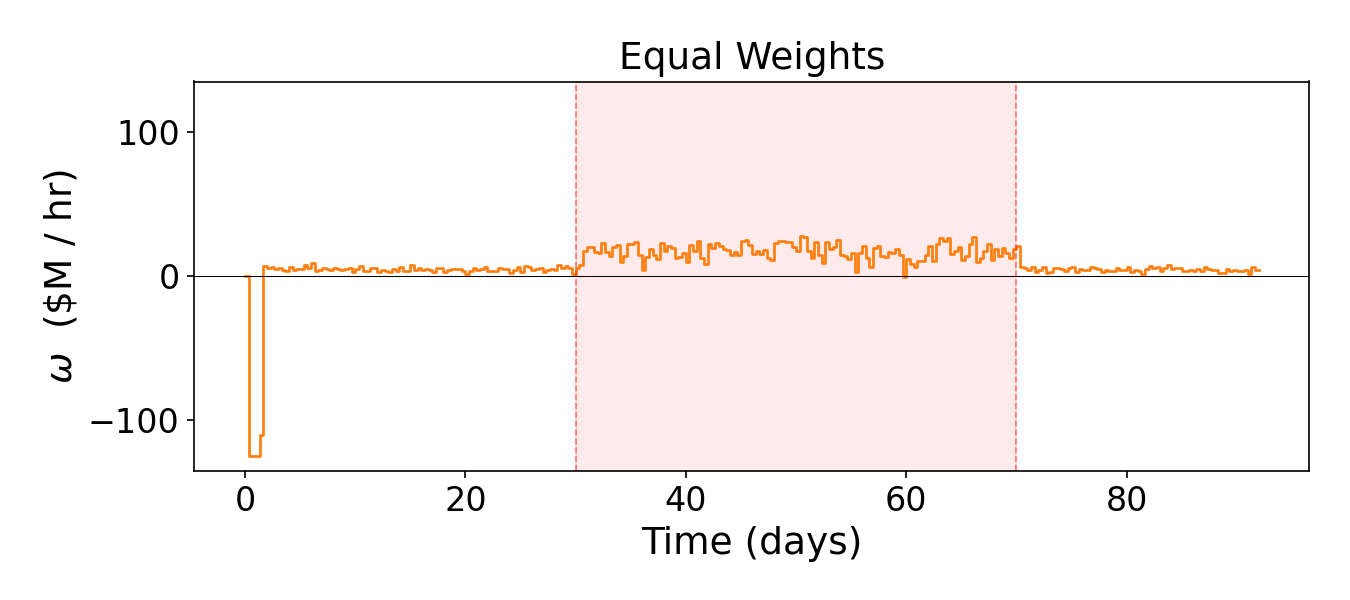}
  \end{tabular}
  \caption{Reallocation rate $\omega(t)$ under the prolonged-clustering
           scenario for all four strategies.}
  \label{fig:controls-grid-clustering}
\end{figure}

\begin{figure}[htbp]
  \centering
  \begin{tabular}{@{}cc@{}}
    \includegraphics[width=0.49\textwidth]{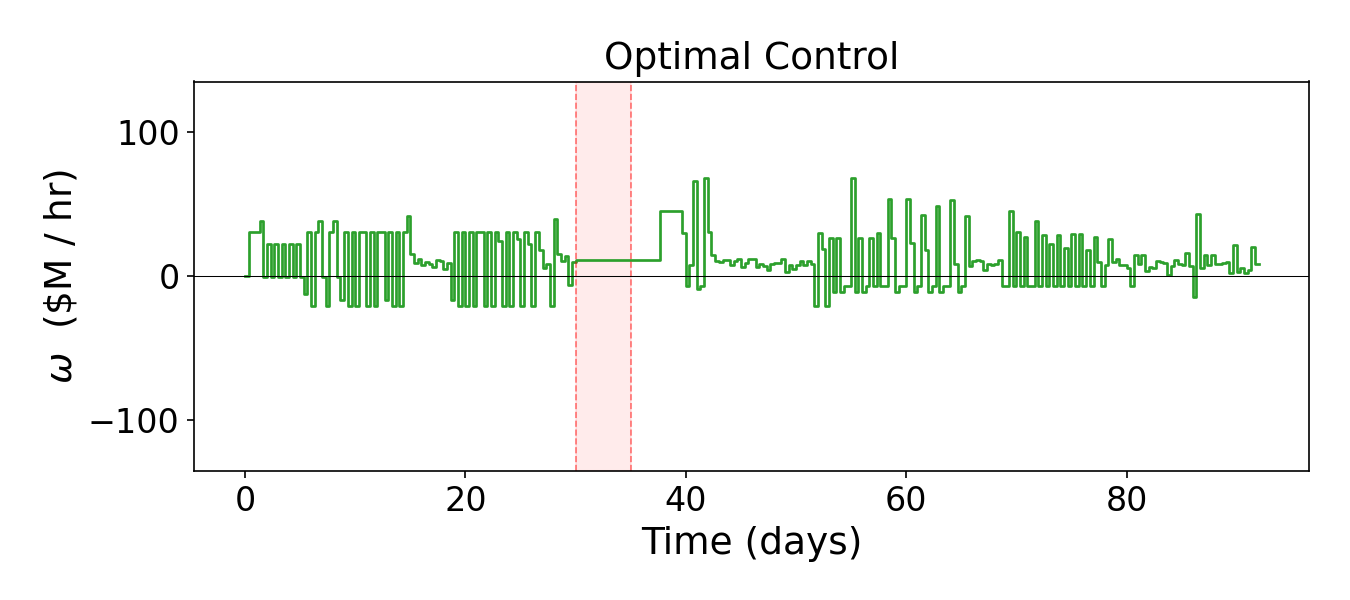} &
    \includegraphics[width=0.49\textwidth]{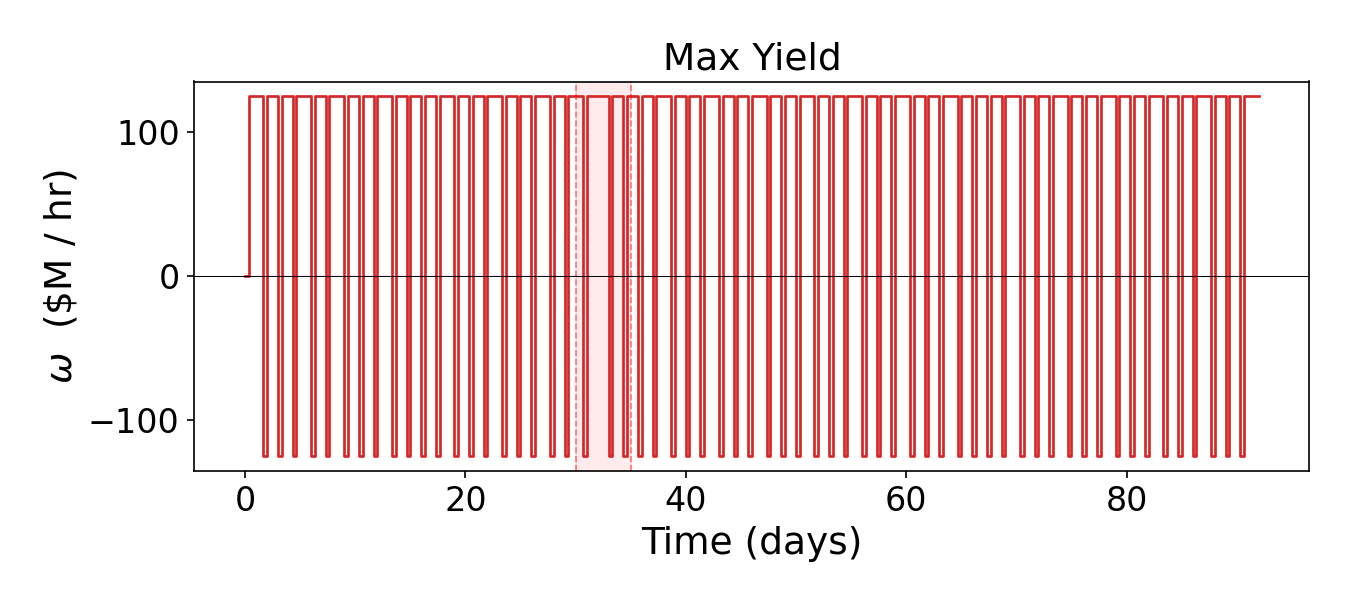} \\[4pt]
    \includegraphics[width=0.49\textwidth]{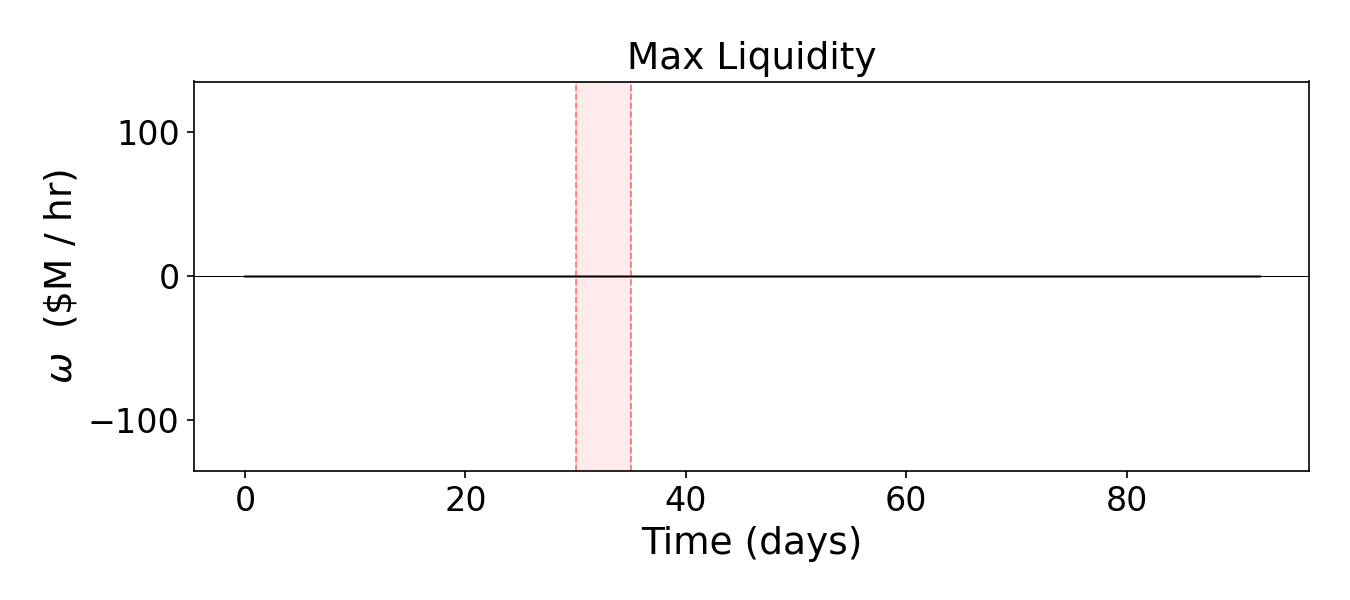} &
    \includegraphics[width=0.49\textwidth]{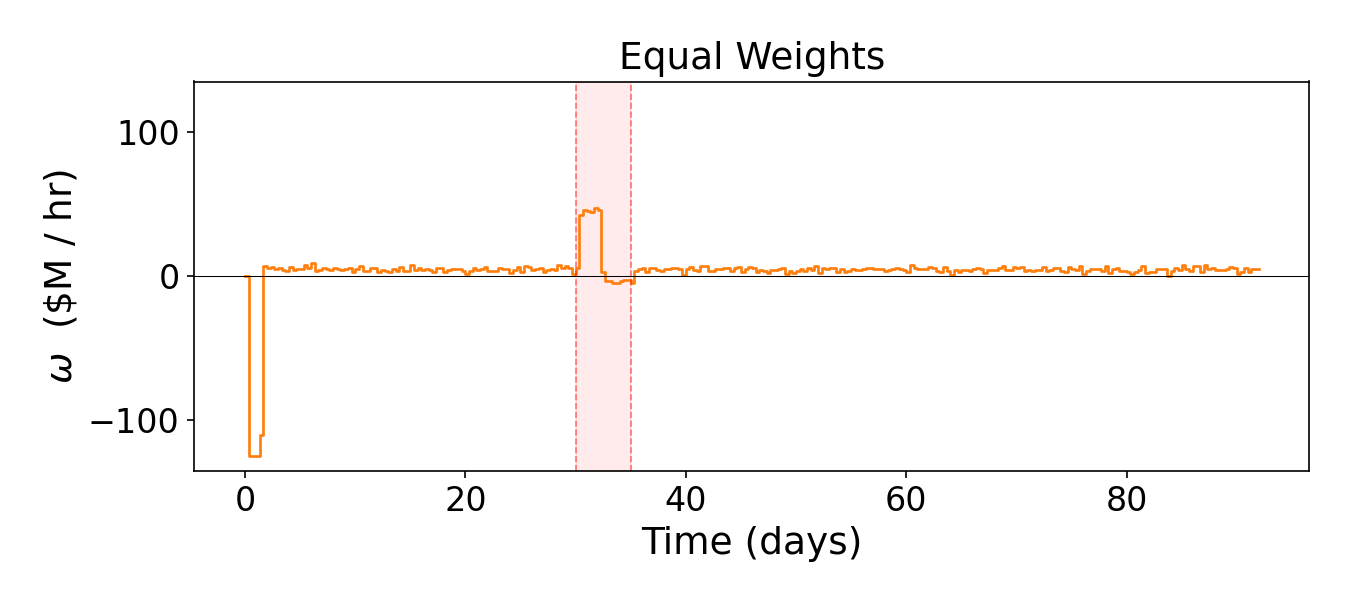}
  \end{tabular}
  \caption{Reallocation rate $\omega(t)$ under the false-alarm scenario for all four strategies.}
  \label{fig:controls-grid-fa}
\end{figure}

The simulation results demonstrate the effectiveness of the proposed control strategy across all three stress scenarios. The optimal controller adapts to changing market regimes without prior knowledge of when stress transitions occur, relying solely on filtered intensity estimates and forward-looking costate evaluation. This is visible in Figures \ref{fig:controls-grid}–\ref{fig:controls-grid-fa}: the reallocation rate shifts decisively toward cash building as redemption pressure materializes, and reverts to yield-seeking allocations once intensities subside. By contrast, the three benchmark strategies each exhibit characteristic failure modes that illustrate the limitations of static allocation rules.

The maximum yield strategy proves clearly inferior across all scenarios. Its control trajectory (Figures \ref{fig:controls-grid}–\ref{fig:controls-grid-fa}, upper right panels) reveals persistent high-frequency oscillation between maximum cash-to-bill and bill-to-cash transfers, driven by the strategy's reactive logic of liquidating bills only when cash is immediately insufficient. This oscillation is operationally costly and prevents sustained accumulation of yield-bearing reserves. More critically, aggressive yield building at stress onset delays the cash build-up needed to absorb clustered redemptions: once the accumulated liquidity shortfall triggers self-exciting bank-run dynamics via the peg-redemption feedback loop (equation \eqref{eq:dyn-lamR}), the resulting depeg is severe and irreversible within the simulation horizon (Figures \ref{fig:peg-combined}–\ref{fig:peg-combined-clustering}). The economic consequences are twofold: direct peg-deviation costs and a permanent reduction of total investable reserves through realized redemptions at distressed levels. In the prolonged-clustering scenario, this failure emerges later (around day 50) as the near-critical branching ratio sustains aftershock bursts that eventually overwhelm the depleted cash buffer.

The maximum liquidity strategy eliminates depeg risk entirely but forgoes all bill carry, yielding zero revenue across all scenarios. The equal weights strategy offers a partial improvement over maximum yield, maintaining a 50/50 target that provides a larger initial cash buffer. However, its backward-looking rebalancing rule does not adapt the target allocation to prevailing or anticipated stress conditions. In the single-shock scenario, the controller operates near maximum cash-building rate throughout the stress window, yet the 50/50 target tracks a shrinking reserve base as net redemptions accumulate, so the absolute cash level targeted also declines. The reallocation rate does not exhibit a pronounced downward shift at the depeg onset (Figure \ref{fig:controls-grid}, lower right) precisely because the controller is already at or near its maximum feasible cash-building speed—there is no additional capacity to deploy. This proves insufficient under sustained pressure, resulting in severe depegs in 69\% of single-shock runs. In the prolonged-clustering and false-alarm scenarios, the equal weights strategy avoids depegs, but its fixed allocation rule prevents it from capturing the higher bill carry that the optimal controller achieves during calm phases and during stress episodes where full cash conversion is unnecessary.
The optimal controller is thus the only strategy that simultaneously (i) builds yield-bearing reserves efficiently during calm periods, as evidenced by positive reallocation rates in the pre-stress and post-stress phases of Figures 4–6, and (ii) achieves sufficiently swift and deep reallocation to liquid cash under stress to prevent the self-reinforcing depeg dynamics that afflict both the maximum yield and equal weights strategies. The false-alarm scenario (Figure \ref{fig:peg-combined-fa}) further confirms that this responsiveness does not come at the cost of overreaction: the controller absorbs the transient mint surge and modest redemption increase without unnecessary cash accumulation, preserving carry throughout.

\section{Conclusion}
\label{sec:conclusion}

This paper develops an operations-first real-time framework for managing the reserves of a pegged digital currency, accounting for clustering dynamics in order and redemption flows and an endogenous dynamical relationship between peg deviation and redemption intensities. We represent mints and redemptions as self-exciting point processes and treat peg behavior as an explicit operational outcome by directly linking price deviations to deficits in immediate cash coverage relative to outstanding supply. The coin issuer's control mechanisms comprise two key instruments: the reallocation rate between immediately available cash and short-duration government bills, and the spread applied to mint or burn requests.

Methodologically, the control architecture combines stochastic model predictive control with a moment-closure forecast for Hawkes intensities and first-order conditions from Pontryagin’s maximum principle. This yields interpretable shadow prices for cash, bills, and peg deviation, and in turn produces implementable feedback rules. A settlement-window formulation converts the continuous problem into a sampled-data problem that matches wire cut-offs and repo calendars. In the discretized problem formulation, the rebalancing decision takes a soft-thresholding form driven by window-averaged costates, which provides a clear and auditable switching logic for treasury desks. We also establish a monotone stress-response property: as expected redemption pressure increases, as average ticket sizes grow, or as windows lengthen, the optimal policy tilts predictably toward cash. Conversely, the proposed algorithm enables efficient reallocation into higher-yielding government bonds during calm market phases, if time windows for reallocation decisions are chosen to be appropriately small.

The numerical experiments show that the window-based controller outperforms empirically common treasury heuristics. Relative to (i) a maximum-yield rule that defers liquidity building until the final settlement window and (ii) a maximum-liquidity rule that sacrifices carry, the controller retains a large share of bill carry in benign conditions while pre-emptively building cash as projected redemption intensities rise, thereby reducing de-peg risk.  Because the policy depends only on filtered intensities, outstanding supply, and window-averaged shadow values, it fits daily treasury workflows and supports supervisory review. More broadly, the framework makes the link between balance-sheet choices and market observables explicit: peg deviations can occur only when intra-window cash is insufficient relative to the outstanding supply, and the controller’s reallocation rate and fee schedule manage that risk directly.

Several potential directions for future research emerge from the framework presented in this paper. A critical risk management concern is the tail probability of a complete depeg. Since such an event is close to catastrophic for the issuer, its probability warrants careful measurement even when it is very small. In this article we simulated 100 runs, which provides only limited estimation precision of probabilities below the one percent level. Future analyses can tighten this by running much larger simulation batches on specialized hardware, by applying dedicated rare-event techniques, or by developing analytical tail approximations.

Several extensions, both in terms of theoretical modeling and empirical validation, would deepen applicability of the proposed approach. Richer flow models with state-dependent trade sizes and cross-asset excitation, online estimation for time-varying branching ratios, robust or risk-sensitive MPC variants, and historical backtests with issuer data would allow closer alignment with specific mandates and regulatory constraints. Taken together, the contributions here provide a settlement-aware, interpretable, and auditable control system that turns reserve management into a disciplined optimization problem and offers a practical blueprint for resilient peg maintenance.

\section*{Acknowledgement(s)}
During the preparation of this manuscript, the author used ChatGPT 5 and Claude 4 for stylistic improvements and coding assistance. The outputs were carefully reviewed and revised, and all responsibility for the final content rests with the author.

\section*{Disclosure Statement}
The author reports there are no competing interests to declare.

\section*{Funding}
No funding was received for the preparation of this manuscript.

\bibliographystyle{tfq}         
\bibliography{interacttfqsample}        
\end{document}